\documentclass[prr,twocolumn,superscriptaddress,
 amsmath,amssymb,
 aps,
]{revtex4-2}

\usepackage{graphicx}
\usepackage{dcolumn}
\usepackage{bm}
\newcommand{\hnote}[1]{\textcolor{blue}{[#1]}} 
 
\usepackage{hyperref}
\usepackage{amsthm}
\usepackage{dsfont}
\usepackage{mathtools}
\usepackage{comment}
\usepackage{hyperref}
\usepackage{lipsum}
\hypersetup{
    colorlinks=true,
    citecolor = [rgb]{0.3,0.8,0},
    linkcolor=blue,
    filecolor=magenta,      
    urlcolor=[rgb]{0.4,0.7,0},
}

\usepackage[capitalise, nameinlink]{cleveref}
\newtheorem*{main result}{Main Result}

\newtheorem{lemma}{Lemma}
\newtheorem{result}{Result}

\DeclarePairedDelimiter{\norm}{\lVert}{\rVert}
\DeclarePairedDelimiter{\snorm}{\lVert}{\rVert_1}
\DeclareMathOperator{\dis}{D}
\usepackage[colorinlistoftodos]{todonotes}

\begin{document}

\title{Robust Semi-Device Independent Certification of All Pure Bipartite Maximally Entangled States via Quantum Steering}
\author{Harshank Shrotriya}
\author{Kishor Bharti}
\affiliation{Centre for Quantum Technologies, National University of Singapore}
\author{Leong-Chuan Kwek}
\affiliation{Centre for Quantum Technologies, National University of Singapore} \affiliation{MajuLab, CNRS-UNS-NUS-NTU International Joint Research Unit, Singapore UMI 3654, Singapore}
\affiliation{National Institute of Education, Nanyang Technological University, Singapore 637616, Singapore}


\begin{abstract}
 The idea of self-testing is to render guarantees concerning the inner workings of a device based on the measurement statistics. It is one of the most formidable quantum certification and benchmarking schemes. Recently it was shown in [Coladangelo et al. Nat. Commun. 8, 15485 (2017)] that all pure bipartite entangled states can be self tested in the device independent scenario by employing subspace methods introduced by [Yang et al. Phys. Rev. A 87, 050102(R) (2013)]. Here, we have adapted their method to show that any bipartite pure entangled state can be certified in the semi-device independent scenario through Quantum Steering. Analogous to the tilted CHSH inequality, we use a steering inequality called Tilted Steering Inequality for certifying any pure two-qubit entangled state. Further, we use this inequality to certify any bipartite pure entangled state by certifying two-dimensional sub-spaces of the qudit state by observing the structure of the set of assemblages obtained on the trusted side after measurements are made on the un-trusted side. As a novel feature of quantum state certification via steering, we use the notion of Assemblage based robust state certification to provide robustness bounds for the certification result in the case of pure maximally entangled states of any local dimension.
\end{abstract}

\maketitle


\noindent {\em Introduction.---} 
Quantum certification and benchmarking become tasks of paramount importance as we advance towards the second quantum revolution \cite{eisert2020quantum}.  One of the prominent approaches to device certification is self-testing \cite{popescu1992generic, mayers_yao, st_review}, the idea of self-testing is to provide guarantees regarding the inner working of a device based on the measurement statistics.  Self-testing  is an important task from the point of view of both quantum foundations research \cite{robust_qdc, yang_navascues,sos_chsh, rst_multipartite} as well as many quantum information processing tasks such as randomness generation \cite{st_random, st_random_2, st_random_3}, quantum cryptography \cite{mayers_yao_2, mayers_yao, jain2017parallel} and entanglement certification \cite{st_entanglement_certification_1, st_entanglement_certification_2}. Self-testing as a method for certification of quantum states and measurements originated in the Device Independent (DI) scenario \cite{popescu1992generic, mayers_yao, mayers_yao_2} taking advantage of Bell Nonlocality \cite{bell_nonlocality_review} to make the claims for certification.  Since then, many important results have been achieved for device-independent self-testing such as self-testing of all pure bipartite entangled states \cite{coladangelo2017all}, self-testing of all multipartite entangled states that admit a Schmidt decomposition \cite{supic_multipartite}, sequential \cite{sequential_st} and parallel self-testing of many EPR pairs \cite{mckague_st_parallel,parallel_selftesting} and GHZ states \cite{parallel_ST_ghz}. Apart from the DI scenario, other scenarios have been explored such as self-testing local quantum systems using non-contextuality inequalities \cite{noncontexuality_selftesting,bharti2019local}, self testing in the prepare and measure scenario \cite{st_pm}, self testing of quantum circuits \cite{st_q_circuits} and semi-device independent (SDI) state certification based on EPR steering \cite{goswami2018one, self_testing_epr, gheorghiu2017rigidity}. In \cite{self_testing_epr}, authors are mainly concerned with showing improvements in robust state certification in the SDI scenario compared to DI scenario while in \cite{goswami2018one}, the authors have shown one sided device independent state certification of any pure two-qubit entangled state. Since the "self testing" terminology was originally defined for the DI scenario, we refrain from using it in this paper and instead refer to our approach as SDI state certification.

Quantum steering \cite{epr, schrodinger_1935, schrodinger_1936} as a phenomenon is logically different from entanglement and nonlocality, lying midway between them. Nonlocal states can be shown to be steerable and steerable states can be shown to be entangled; however, nonlocality is studied in the DI scenario while steering is studied in the SDI scenario. Applications of quantum steering have been investigated in Quantum Key Distribution (QKD) \cite{qkd_steering} using the BBM92 protocol \cite{bbm92}, where the authors showed that in the SDI scenario obtainable key rates are higher and the required detector efficiencies are lower compared to the fully DI case. Steering based randomness certification has been studied in \cite{law2014quantum, passaro2015optimal, randomness_certification_steering}. In \cite{randomness_certification_steering}, the authors show maximal randomness generation in the SDI setting from any pure entangled full-Schmidt-rank state using the maximal violation of steering inequalities. See \cite{steering_review} for a comprehensive review on quantum steering and \cite{cavalcanti_review} for a review based on semi-definite programming.

Self-testing in the DI scenario entails obtaining specific extremal correlations while treating the measurement devices as black boxes. Such extremal correlations violate inequalities such as the CHSH inequality \cite{chsh} maximally. Although SDI is a weaker notion and requires more assumptions than the fully DI case but the advantage of state certification in the SDI scenario are manifold; the additional assumptions help in alleviating mathematical difficulties such as establishing a tensor product structure in the multipartite case, or certification of complex measurements as described in \cite{self_testing_epr}, state certification using the maximal violation of a steering inequality in the SDI scenario was shown to be advantageous in a laboratory setting \cite{cavalcanti2009experimental}. Further, as we show in this paper, steering based approach renders robust state certification tractable by simplifying the mathematical analysis required. Such robustness bounds were missing in previous DI self testing result \cite{coladangelo2017all} which makes the steering based approach favourable for experimental implementations \cite{exp_SDI_selftest, exp_steering, saunders2010experimental, malik2021genuine}. Very recently, \cite{manvcinska2021constant, fu2019constant} obtained robust self testing results using correlations for different subsets of maximally entangled states; however, the correlations based approach is non-constructive in nature while in the steering case we have shown explicit bounds.

The assumptions required in SDI scenario are naturally justified, such as in the case of delegated quantum computing \cite{blind_qc}. In \cite{gheorghiu2017rigidity}, the authors discuss rigidity of quantum steering correlations via sequential steering games and show that the overhead is reduced in the steering case compared to that of CHSH game rigidity. The authors further highlight the application of steering rigidity in verifying delegated quantum computation.

Moved by the advantages of the SDI scenario, it is natural to enquire whether all pure bipartite entangled states can be certified in the SDI scenario via quantum steering. Progress in the aforementioned direction has been made in \cite{self_testing_epr,gheorghiu2017rigidity} where the authors used a linear steering inequality for certifying the maximally entangled Bell pair. One-sided device-independent certification of any two-qubit pure entangled state was shown in \cite{goswami2018one} where the authors used two steering inequalities, Fine-grained inequality (FGI) \cite{fine_grained_ineq} and analogue CHSH \cite{cavalcanti2015analog} inequalities for self-testing. However, the maximal violation of the FGI is not uniquely achieved by a particular target state, and the authors used the analogue CHSH inequality, along with maximal violation of FGI, to uniquely ascertain which state had been certified.

In this work we have used the Tilted Steering Inequality, the maximal violation of which uniquely certifies any two-qubit pure entangled state. Further, by adapting the subspace methods of \cite{coladangelo2017all} for the SDI scenario, we provide a  state certification scheme for any pure bipartite entangled state via quantum steering. Finally, we provide $O(\sqrt{\epsilon})$ robustness bounds for our steering based state certification result for maximally entangled pure qudit states which are not known in the device independent scenario.\\

\noindent {\em Steering Scenario.---}
We have considered an SDI scenario where Alice and Bob share an entangled qudit state, such that Alice performs black-box measurements on her side and Bob performs tomographic reconstruction to obtain the exact density matrix of the states on his side after Alice's measurements. 

Quantum steering was first formalised in \cite{wiseman_2007} where the authors defined steerable states as those which do not admit a Local Hidden State (LHS) model for any assemblage generated on Bob's side. The assemblage (un-normalised state) $\sigma_{a|x}$ generated on Bob's side corresponding to outcome $a$ and measurement $x$ on Alice's side is said to admit an LHS model if it can be represented as a mixture of hidden states $\rho_\lambda$ originating from a probability distribution $\mu(\lambda)$:
\begin{equation}
    \sigma_{a|x}=\int d\lambda \mu(\lambda)p(a|x,\lambda)\rho_\lambda
\end{equation}
In accordance with quantum mechanics, the assemblage generated by performing projective measurements $M_{a|x}$ on Alice's side can be written as:
\begin{equation}
    \sigma_{a|x}=tr_A[(M_{a|x}\otimes I)\rho^{AB}]
\end{equation}
where $\sum_a M_{a|x}=\mathds{1}$ and $M_{a|x}\geq 0\, \forall a,x$. In \cite{cavalcanti2015analog}, the authors described steerable states by modelling non-steerable correlations using a Local Hidden Variable-Local Hidden state (LHV-LHS) model. Given a bipartite system comprising of qubits with spatially separated parties Alice and Bob, denote $O_A$ and $O_B$ as the set of observables in the Hilbert space of Alice and Bob respectively. An element in $O_A$ will be denoted by $x$ (similarly $y$ for $O_B$) and the outcomes corresponding to $x$ will be denoted by $a\in L(x)$ (similarly $b\in L(y)$) where $L(x)$ $(L(y))$ denotes the set of outcomes for the observable $x$ $(y)$. The joint state $\rho^{AB}$ is said to be steerable if and only if it does not admit an LHV-LHS decomposition for \textbf{all} $a\in L(x)$, $b\in L(y)$, $x\in O_A$ and $y\in O_B$. An LHV-LHS decomposition is based on the idea that Alice's outcomes are determined by a local hidden variable $\lambda$ and Bob's outcomes are determined by local measurements on a quantum state $\rho_\lambda$,
\begin{equation}
    p(a,b|x,y;\rho^{AB})=\sum_\lambda \mathfrak{p}(\lambda)\mathfrak{p}(a|x,\lambda) p(b|y,\rho_\lambda) 
\end{equation}
In the SDI scenario, the LHV-LHS model can be used to establish local bounds for linear expressions giving rise to steering inequalities, violation of such inequalities implies steering (see Appendix Section \ref{local_bound_TSI}). \textit{Steerable Weight} as a quantifier of steering was proposed in \cite{steering_weight} and was further shown to be a convex steering monotone in \cite{resource_steering}. Consider a one-sided device-independent scenario where Alice and Bob share a joint quantum state $\rho^{AB}$, an arbitrary set of assemblages $\{\sigma_{a|x}\}_{a,x}$ (set of assemblages obtained overall outcomes and observables) can be decomposed as,
\begin{equation}
    \sigma_{a|x}=p_s\sigma^S_{a|x}+(1-p_s)\sigma^{US}_{a|x} \quad\forall a,x
\end{equation}
where $0\leq p_s\leq 1$, $\sigma^S_{a|x}$ is a steerable assemblage and $\sigma^{US}_{a|x}$ is an un-steerable assemblage i.e. has an LHS decomposition. The weight of the steerable part $p_s$ minimized over all possible decompositions of $\{\sigma_{a|x}\}_{a,x}$ gives the steerable weight $SW(\{\sigma_{a|x}\}_{a,x})$ of that assemblage set. \\

\noindent {\em SDI state certification---}
Self testing was originally introduced in the DI scenario where the maximal violation of the CHSH inequality was used to self test the Bell state \cite{mayers_yao, mayers_yao_2}. Such self testing procedures are based on obtaining extremal correlations $p(a,b|x,y)=tr[(M_{a|x}\otimes N_{b|y})\rho^{AB}]$, obtained by performing quantum measurements $M_{a|x}$ (acting on $\mathcal{H}_A$) and $N_{b|y}$ (acting on $\mathcal{H}_B$) on the joint quantum state  $\rho^{AB}\in \mathcal{H}_A\otimes \mathcal{H}_B$, such that the correlations $p(a,b|x,y)$ achieve the quantum supremum of a Bell inequality and the quantum states that achieve the extremal correlations are unique up to local isometries. Formally defined, the extremal correlations $p(a,b|x,y)$ self test the state and measurements  $\{|\Bar{\psi}\rangle,\Bar{M}_{a|x},\Bar{N}_{b|y}\}$ if for all states and measurements $\{|\psi\rangle,M_{a|x},N_{b|y}\}$ compatible with $p(a,b|x,y)$, there exists an isometry $\Phi=\Phi_A\otimes\Phi_B$ where $\Phi_A:\mathcal{H}_A \mapsto \mathcal{H}_A\otimes \mathcal{H}_{A'}$ and $\Phi_B:\mathcal{H}_B \mapsto \mathcal{H}_B\otimes \mathcal{H}_{B'}$ such that:
\begin{equation}
    \begin{aligned}
    \Phi(|\psi\rangle_{AB})&=|junk\rangle_{AB}|\Bar{\psi}\rangle_{A'B'} \\
    \Phi(M_{a|x}\otimes N_{b|y}|\psi\rangle_{AB})&=|junk\rangle_{AB}(\Bar{M}_{a|x}\otimes \Bar{N}_{b|y}|\Bar{\psi}\rangle_{A'B'})
    \end{aligned}
\end{equation}
where $\Bar{M}_{a|x}$ and $\Bar{N}_{b|y}$  act on $\mathcal{H}_{A'}$ and $\mathcal{H}_{B'}$ respectively.\\
Coming to the SDI scenario, for certifying any pure bipartite entangled state, using Schimdt decomposition, our target state can be written as:
\begin{equation}
\label{qudit_state}
    |\psi_{target}\rangle:=\sum_{i=0}^{d-1}c_i|ii\rangle
\end{equation}
where $0<c_i<1 \, \forall\,i$ and $\sum_{i=0}^{d-1}c_i^2=1$.
In order to certify the target state \eqref{qudit_state}, we intend to show the existence of an isometry $\Phi:\mathcal{H}_A \mapsto \mathcal{H}_A\otimes \mathcal{H}_{A'} $ on Alice's side such that:
\begin{equation}
\label{isometry}
    \begin{aligned}
    \Phi(|\psi\rangle_{AB})&=|junk\rangle_A\otimes|\psi_{target}\rangle_{A'B} \\
    \Phi(M_{a|x}|\psi\rangle_{AB})&=|junk\rangle_A\otimes \Bar{M}_{a|x}|\psi_{target}\rangle_{A'B}
    \end{aligned}
\end{equation}
where $M_{a|x}$ acts on $\mathcal{H}_{A}$; $\Bar{M}_{a|x}$ acts on $\mathcal{H}_{A'}$ and represents ideal measurements on Alice's side. Formally stated, this result is the following.
\begin{result}
\label{main_result}
In the SDI scenario, for any bipartite pure entangled state $|\psi_{target}\rangle$, there exists an ideal set of assemblages $\{\sigma_{a|x}\}_{a,x}^{ideal}$, where $x\in \{0,1,2\}$ and $a\in \{0,\ldots,d-1\}$, which when observed on Bob's side after Alice performs black box measurements on their joint state $\rho^{AB}$ certifies the target state $|\psi_{target}\rangle$ and ideal measurements on Alice's side.
\end{result}

\begin{figure}[htbp]
  \includegraphics[width=\columnwidth,height=7.5cm]{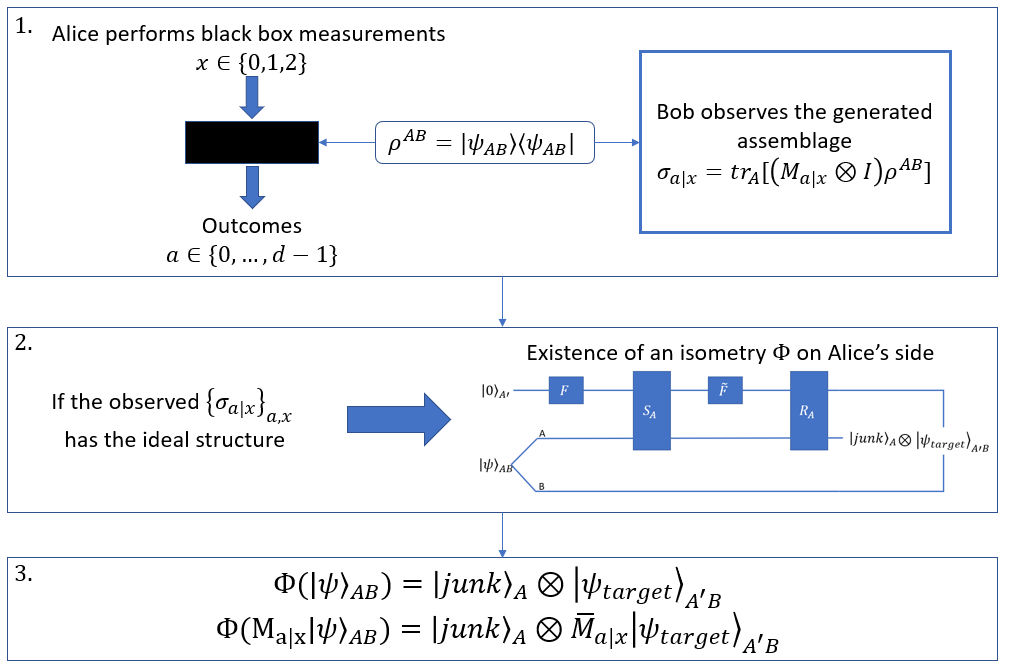}
  \caption{SDI state certification scheme: Step 1. Alice performs uncharacterized measurements on her side with three measurement choices $x\in\{0,1,2\}$ and records the outcomes $a\in\{0,\ldots,d-1\}$. For each measurement $x$ and outcome $a$, Bob measures the assemblage generated on his side $\sigma_{a|x}$. The experiment is repeated several times for each measurement setting $x$ to build the complete assemblage set on Bob's side. Step 2. If the observed assemblage set $\{\sigma_{a|x}\}_{a,x}$ has the ideal structure (see Appendix Section \ref{assemblage_structure}) then the isometry $\Phi$ (see Fig. \ref{Alice_isometry}) can be shown to exist on Alice's side. Step 3. The isometry $\Phi$ certifies the target state \eqref{qudit_state} and the ideal measurements as shown in \autoref{sufficient_lemma} in Appendix Section \ref{SDI_proof_ST}.}  
  \label{st_scheme}
\end{figure}
\noindent {\em Proof Sketch.---}
Here we provide a sketch of the proof for this result, the detailed proof is given in Appendix \ref{SDI_proof_ST}. Since our aim is to certify the state \eqref{qudit_state} by certifying the two dimensional subspace projected blocks, we first establish the Tilted Steering Inequality (TSI) analogous to the tilted CHSH inequality \cite{tiltedCHSH}, the maximal violation of which uniquely certifies any two-qubit pure entangled state $|\psi(\theta)\rangle=\cos\theta|00\rangle+\sin\theta|11\rangle$ in the SDI scenario. The TSI is a two parameter inequality given by:
\begin{equation}
\label{TSI_MT}
    I_{\alpha,\beta}\equiv \alpha\langle A_0\rangle+\beta\langle A_0 Z_s\rangle+\langle A_1 X_s\rangle\leq\alpha+\sqrt{1+\beta^2}
\end{equation}
where $\beta, \alpha>0$, $A_0$ and $A_1$ refer to 2-outcome black-box measurements on Alice's side and $Z_s=\frac{|0\rangle \langle0| - |1\rangle \langle1|}{2}$ and $X_s=\frac{|0\rangle \langle1| + |1\rangle \langle0|}{2}$ refer to local schmidt basis measurements performed by Bob. In order to certify a given bipartite pure entangled state $|\psi(\theta)\rangle$, we require 2 additional constraints i.e. $\sin2\theta = 1/\beta$ and impose a condition on the parameters $\alpha$ and $\beta$ such that $\beta^2 = \alpha^2 + 1$. Under these 2 conditions the state $|\psi(\theta)\rangle$ maximally violates TSI uniquely and this maximal violation is achieved when Alice's local schmidt basis measurements are parallel to that of Bob's i.e. Alice makes $A_0=\frac{|0\rangle \langle0| - |1\rangle \langle1|}{2}$ and $A_1=\frac{|0\rangle \langle1| + |1\rangle \langle0|}{2}$ measurements.

In order to calculate the local and quantum bounds for the linear expression $I_{\alpha,\beta}$, we use the concepts defined earlier, namely the LHV-LHS model and steerable weight. In the LHV-LHS model, Alice's measurement outcomes are determined using a probabilistic LHV model and Bob's outcomes come from local measurements on a quantum state. We calculate the local bound for TSI within the LHV-LHS model by individually maximising the terms of $I_{\alpha,\beta}$ over the allowed set of outcome probabilities for Alice and Bob. 
For calculating the quantum bound, we first consider a property of steerable weight ($SW(\sigma_{a|x})$) given in \cite{cavalcanti_review}; $SW(\sigma_{a|x})$ is bounded by any convex function $f(.)$ of that assemblage, where $f(.)$ can be the violation of a steering inequality. Using this property, it can be further shown that:
\begin{equation}
\label{SW_prop}
    SW(\sigma_{a|x}) \geq \frac{f(\sigma_{a|x})-f^{LHS}_{max}}{f_{max}-f^{LHS}_{max}},
\end{equation}
where $f_{max}$ and $f^{LHS}_{max}$ are the maximal value of $f(.)$ among all possible assemblages and among all possible LHS assemblages respectively. From the above equation it can be inferred that SW of an assemblage that gives the maximal violation of a steering inequality must be greater than or equal to 1, which is in fact the maximum value that SW can take. Hence, SW of the assemblage that achieves the maximal violation must be 1. Then we show that the quantum bound cannot be achieved by an assemblage generated from a mixed two-qubit entangled state using \eqref{SW_prop} and a lemma from \cite{goswami2018one}. Having eliminated mixed entangled states, we maximise the expression $\alpha\langle A_0\rangle+\beta\langle A_0 B_0\rangle+\langle A_1 B_1\rangle$ over all general projective measurements performed by Alice and Bob on a general pure two-qubit entangled state $|\psi(\theta)\rangle=\cos\theta|00\rangle+\sin\theta|11\rangle$ where $0<\theta<\pi/2$ thus giving the quantum bound $\sqrt{2(1 + \alpha^2 + \beta^2)}$. Kindly note that since we fix the local bases to be the Schmidt bases we cannot write the optimal Bob's measurements as $Z$ and $X$. However, the local bases can always be rotated to give measurement directions as $Z$ and $X$.

We show explicitly in Appendix \ref{quantum_bound_TSI} that for all values of $\alpha, \beta >0$, the quantum bound is always greater than or equal to the local bound; however, for the purpose of certifying the state $|\psi(\theta)\rangle$ we impose an additional condition $\beta^2 = \alpha^2 + 1$ on the parameters such that the state $|\psi(\theta)\rangle$ achieves the maximal violation $\sqrt{2(1 + \alpha^2 + \beta^2)}$ which is the quantum bound. The detailed proofs for obtaining the local and quantum bounds are given in Appendix \ref{TSI}.

Secondly we look for a laboratory fingerprint of the qudit state which can be used to show the existence of an isometry, as are the correlations in the DI case; naturally this is achieved by using assemblages in the SDI scenario where we impose a $2\times2$ outer product structure on the ideal  assemblages generated on Bob's side for certain measurements of Alice (see Appendix Section \ref{assemblage_structure} for details). For certifying a general state of the form \eqref{qudit_state}, we show that Alice needs to perform 3 $d$-outcome measurements on her side; the structure imposed on the assemblage set $\{\sigma_{a|x}\}_{a,x}^{ideal}$, where $a\in\{0,\ldots,d-1\}$ and $x\in\{0,1,2\}$, is such that for measurement settings $x\in\{0,1\}$ the pairs $c_{2m}|2m,2m\rangle+c_{2m+1}|2m+1,2m+1\rangle$ and for measurement settings $x\in\{0,2\}$ the pairs $c_{2m+1}|2m+1,2m+1\rangle+c_{2m+2}|2m+2,2m+2\rangle$ for $m=\{0,\ldots,\frac{d}{2}-1\}$ are certified respectively. The intuition behind certifying 2 different sets of pairs, similar to the DI case \cite{coladangelo2017all}, is that the maximal violation of tilted steering inequality certifies the corresponding 2 dimensional normalised projections $|\psi_m\rangle=\frac{c_{2m}|2m,2m\rangle+c_{2m+1}|2m+1,2m+1\rangle}{\sqrt{c_{2m}^2+c_{2m+1}^2}}$ or $|\psi_m\rangle=\frac{c_{2m+1}|2m+1,2m+1\rangle+c_{2m+2}|2m+2,2m+2\rangle}{\sqrt{c_{2m+1}^2+c_{2m+2}^2}}$ of the target state \eqref{qudit_state}. That is, the maximal violation of a particular $I_{\alpha,\beta}$ certifies a state of the form $|\psi(\theta)\rangle$ satisfying $\sin2\theta=\frac{1}{\beta}$ which certifies the ratio $\tan\theta=\frac{c_{2m+1}}{c_{2m}}$ or $\tan\theta=\frac{c_{2m+2}}{c_{2m+1}}$ of the coefficients. Hence, we obtain $d/2$ relations between the coefficients from the set corresponding to measurements $x\in\{0,1\}$ and another $d/2$ from the set corresponding to $x\in\{0,2\}$ thus uniquely determining the $d$ coefficients of the target state \eqref{qudit_state}.\\

\begin{figure}
  \includegraphics[width=\columnwidth]{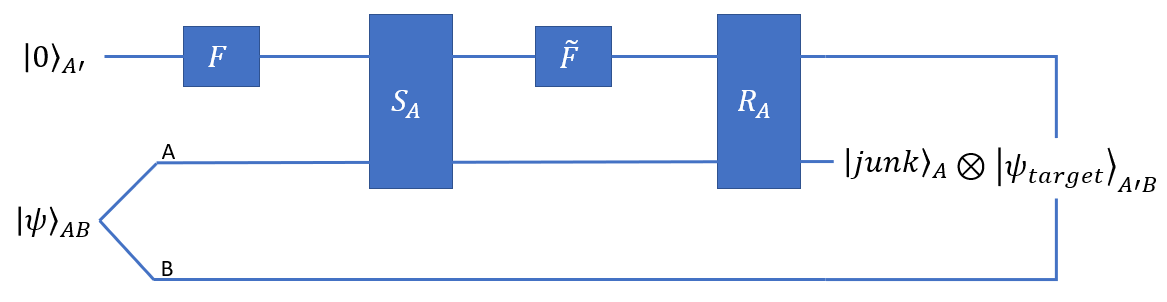}
  \caption{Circuit diagram for the local isometry $\Phi$ on Alice's side. Gates $F$ and $\Bar{F}$ are quantum Fourier transform and inverse quantum Fourier transform respectively. Gate $R_{AA'}$ is defined as $R_{AA'}|\psi\rangle_{AB}|k\rangle_{A'}=X_A^{(k)}|\psi\rangle_{AB}|k\rangle_{A'}$ where $X_A^{(k)}$ are unitary operators and gate $S_{AA'}$ is defined as $S_{AA'}|\psi\rangle_{AB}|k\rangle_{A'}=Z_A^{k}|\psi\rangle_{AB}|k\rangle_{A'}$ where $Z_A:=\sum_{k=0}^{d-1} \omega^k P_A^{(k)}$ and $P_A^{(k)}$ are projections corresponding to measurement $x=0$ (refer to Appendix Section \ref{SDI_proof_ST} for complete descriptions of all the gates).}
  \label{Alice_isometry}
\end{figure}

Finally, to formally show the existence of an isometry \eqref{isometry}, we construct projection operators $\{P_A^{(k)}\}_{k=0,\ldots,d-1}$ and unitaries $X_A^{(k)}$ on Alice's side which arise as a consequence of the structure imposed on the assemblage set $\{\sigma_{a|x}\}_{a,x}$. We further show that these projection operators and unitaries satisfy the condition:
\begin{equation}
    X_A^{(k)}P_A^{(k)}|\psi\rangle=c_k|0,k\rangle \quad \forall k
\end{equation}
which is sufficient for proving the existence of an isometry (Fig. \ref{Alice_isometry}) which certifies the target state and the ideal measurements \eqref{isometry} (See \autoref{sufficient_lemma} in Appendix Section \ref{SDI_proof_ST} for proof).\\

\noindent {\em Assemblage based Robust SDI state certification ---}
The notion of assemblage based robust certification was introduced in \cite{self_testing_epr}, where the authors used their assemblage based strategy to analytically derive better robustness bounds for certifying the singlet state i.e. the maximally entangled 2-qubit state as compared to the bounds that were known for the device independent case. To formalise this notion, the distance measure used is the Trace Distance between 2 quantum states; $\dis(\rho_1, \rho_2)$ which is related to the Schatten-1 norm $\snorm{\rho_1 - \rho_2}$ as $\dis(\rho_1, \rho_2) = \frac{1}{2}\snorm{\rho_1 - \rho_2}$. Having established a notion of closeness of 2 quantum states, we state our RST result as:

\begin{main result}
\label{robust_result}
For the maximally entangled pure qudit state $|\Bar{\psi}_{target}^{max}\rangle = \frac{1}{\sqrt{d}}\sum_{i=0}^{d-1}|ii\rangle$ of dimension $d$ with $\Bar{\rho} = |\Bar{\psi}_{target}^{max}\rangle \langle \Bar{\psi}_{target}^{max}|$, if the experimentally obtained assemblages $\{\sigma_{a|x}\}_{a,x}$ and reduced state $\rho_B$ on Bob's side are close to the reference set of assemblages $\{\Bar{\sigma}_{a|x}\}_{a,x}$ and reduced state $\Bar{\rho}_B$ such that $\lVert \sigma_{a|x} - \Bar{\sigma}_{a|x} \rVert_1 \leq \epsilon \quad \forall \, a,x$ and $\dis(\rho_B, \Bar{\rho}_B) \leq \epsilon$ then:
\begin{widetext}
\begin{gather}
         \dis(|\Phi\rangle\langle \Phi|, \rho_{junk} \otimes \Bar{\rho}_{A'B}) \leq \frac{1}{2}(4d^2\sqrt{\epsilon} + d^3\epsilon\sqrt{\epsilon} + d\epsilon) \\
        \snorm{|\Phi,M_{a|x}\rangle\langle \Phi,M_{a|x}| - \rho_{junk} \otimes (\Bar{M}_{a|x}\otimes I_B)|\Bar{\psi}_{target}^{max}\rangle \langle \Bar{\psi}_{target}^{max}|(\Bar{M}_{a|x}\otimes I_B)} \leq 4(d^2+d)\sqrt{\epsilon} + \epsilon
\end{gather}
\end{widetext}
where $|\Phi\rangle := \Phi|\psi\rangle_{AB}|0\rangle_{A'}$, $|\Phi,M_{a|x}\rangle := \Phi(M_{a|x}\otimes I_B)|\psi\rangle_{AB}|0\rangle_{A'}$ and $|\psi\rangle_{AB}$ is the purification of the reduced state $\rho_B$. (See Appendix $\ref{robustness_appendix}$ for detailed derivation)
\end{main result}

\noindent {\em Discussion and Open Problems.---}
In this paper, we have shown SDI certification of any pure bipartite entangled state by certifying two-dimensional subspace projections of the qudit state using a tilted steering inequality adapting the subspace methods previously employed in \cite{coladangelo2017all, yang_navascues}. More importantly we derived robustness bounds for this certification in the case of the maximally entangled pure qudit state of dimension $d$, such explicit robust bounds are not known in the device independent scenario. As shown in Appendix Section \ref{ideal_measurements}, in the ideal case Alice needs to make 3 measurements on her side in order to obtain the set of ideal assemblages which is sufficient for the certification task. Also, it can be easily seen that in the SDI scenario the same set of ideal measurements can be used to certify any qudit state of a given dimension. However, the ideal set of assemblages that we have used to show the existence of an isometry is not necessarily the only set of assemblages that certifies a given bipartite pure entangled state. It has been shown in \cite{steering_weight} that for a set of assemblages generated by performing 2 $d$-outcome projective measurements on any pure bipartite entangled state, a single linear steering inequality can be obtained which is maximally violated by that set of assemblages. Results in this direction could be used to show certification of a bipartite pure entangled state using a single steering inequality obtained using SDP methods, furthermore such a procedure could be shown to require only 2 measurements on Alice's side.

\emph{Acknowledgement. ---}
The authors thank Debarshi Das for providing an improved and complete proof for Result \ref{Result_1_2qubit} given in Appendix \ref{ST_2qubit} which has been incorporated in the draft.
Additionally, the authors would like to thank the National Research Foundation and the Ministry of Education, Singapore for financial support.

\onecolumngrid
\appendix

\section{Tilted Steering Inequality}
\label{TSI}
In this section we will show detailed calculations for obtaining the local and quantum bounds for the tilted steering expression $I_{\alpha,\beta}\equiv \alpha\langle A_0\rangle+\beta\langle A_0Z\rangle+\langle A_1X\rangle$. For obtaining the local bound we will maximise $I_{\alpha,\beta}$ for the correlations obeying the LHV-LHS model and for obtaining the quantum bound (under the condition $\beta>\alpha>0$) we will consider arguments involving steerable weight and entangled states. Further, we will show that the maximal violation of TSI uniquely certifies any 2-qubit pure entangled state.

\subsection{Local bound of Tilted Steering Inequality}
\label{local_bound_TSI}
Let us first recap the LHV-LHS model here for completeness. Given a bipartite system comprising of qubits with spatially separated parties Alice and Bob, denote $O_A$ and $O_B$ as the set of observables in the Hilbert space of Alice and Bob respectively. An element in $O_A$ will be denoted by $x$ (similarly $y$ for $O_B$) and the outcomes corresponding to $x$ will be denoted by $a\in L(x)$ (similarly $b\in L(y)$) where $L(x)$ $(L(y))$ denotes the set of outcomes for the observable $x$ $(y)$.
The joint state $\rho^{AB}$ is said to be steerable iff it does not admit a LHV-LHS decomposition for \textbf{all} $a\in L(x)$, $b\in L(y)$, $x\in O_A$ and $y\in O_B$. A LHV-LHS decomposition is based on the idea that Alice's outcomes are determined by a local hidden variable $\lambda$ and Bob's outcomes are determined by local measurements on a quantum state $\rho_\lambda$,
\begin{equation}
    P(a,b|x,y;\rho^{AB})=\sum_\lambda \mathfrak{p}(\lambda)\mathfrak{p}(a|x,\lambda) p(b|y,\rho_\lambda) \label{lhvlhs}
\end{equation}
For a given scenario (2 measurement 2 outcome in this case) the correlations that have a LHV-LHS model form a convex set \cite{cavalcanti2009experimental}, hence we can express any LHV-LHS model as a convex combination of extremal points of that convex set. This implies that we can decompose $\mathfrak{p}(a|x,\lambda) p(b|y,\rho_\lambda)$ into $\sum_\chi \int d\xi\mathfrak{p}(\chi,\xi|\lambda)\delta_{a,f(A,\chi)}P(b|y;|\psi_\xi\rangle\langle\psi_\xi|)$, where $\chi$ are the variables which determine the extremal outcome strategies for Alice via the function $f(A,\chi)$, and $\xi$ determines a \textbf{pure} state $|\psi_\xi\rangle$ for Bob. Thus we can simplify \eqref{lhvlhs} as:
\begin{equation}
    P(a,b|x,y)=\sum_{\chi,\xi}\mathfrak{p}(\chi,\xi)\delta_{a,f(x,\chi)}\langle\psi_\xi|\Pi_{b|y}|\psi_\xi\rangle. \label{lhvextreme}
\end{equation}

To calculate the local bound of our Tilted steering inequality using the LHV-LHS model \eqref{lhvlhs} where Alice and Bob each have a choice between 2 dichotomic observables given by: $\{x=0,x=1\},\{y=0,y=1\}$. As shown in equation \eqref{lhvextreme} we can express any joint probability distribution obeying the LHV-LHS model as convex sum of extremal points in that set. We will briefly review the method to obtain these extremal points for Alice's and Bob's side as illustrated in \cite{cavalcanti2015analog}. For Alice's side with 2 observables $x=0,x=1$ with outcomes $\pm1$, the 4 deterministic strategies are given by,
\begin{equation}
\label{Aextreme}
    p^{1|0}=1,p^{1|1}=1;\quad p^{1|0}=1,p^{1|1}=0;\quad p^{1|0}=0,p^{1|1}=1;\quad p^{1|0}=0,p^{1|1}=0;
\end{equation}
and let us attach these 4 strategies with 4 labels $\chi\in \{1,2,3,4\}$. Moving ahead to Bob's side, since Bob's probabilities arise from measurements on a quantum state they have to obey constraints such as uncertainty relations. Labelling the basis of eigenstates Bob's observable $y=0$ as $\{|1\rangle, |-1\rangle\}$, we can write the generalised projector for outcome $1$ of observable $y=1$, parameterized by $\mu$ and $\phi$, as:
\begin{equation}
    \label{obsB'}
    \Pi_{1|1}=(\sqrt{\mu}|1\rangle+\sqrt{1-\mu}e^{i\phi}|-1\rangle)\times(\sqrt{\mu}\langle1|+\sqrt{1-\mu}e^{-i\phi}\langle-1|).
\end{equation}
Similarly a general pure state can be written as:
\begin{equation}
    |\psi_{\mu',\phi'}\rangle=\sqrt{\mu'}|1\rangle+\sqrt{1-\mu'}e^{i\phi'}|-1\rangle;
\end{equation}
then the measurement probabilities of outcome $1$ for $y=0$ and $y=1$ are:
\begin{equation}
    \begin{aligned}
    p^{1|0}(\mu',\phi') &\equiv \langle\psi_{\mu',\phi'}|\Pi_{1|0}|\psi_{\mu',\phi'}\rangle=\mu',\\ 
    p^{1|1}(\mu',\phi') &= \mu\mu'+(1-\mu)(1-\mu')+2\sqrt{\mu(1-\mu)\mu'(1-\mu')}\cos(\phi'-\phi). 
    \end{aligned}
\end{equation}
From the above equations, we see that the set of allowed quantum probabilities $(p^{1|0},p^{1|1})$ form the convex hull of an ellipse and the boundaries (extreme values) are achieved when $\cos(\phi'-\phi)=\pm1$. For the extreme points, we can reparameterize the ellipse as:
\begin{equation}
    \label{Bextreme}
    \begin{aligned}
    p^{1|0}(\xi)-\frac{1}{2} &= \frac{1}{2}\lbrack\sqrt{\mu}\cos(\xi)-\sqrt{1-\mu}\sin(\xi)\rbrack,\\ 
    p^{1|1}(\xi)-\frac{1}{2} &=  \frac{1}{2}\lbrack\sqrt{\mu}\cos(\xi)+\sqrt{1-\mu}\sin(\xi)\rbrack.
    \end{aligned}
\end{equation}
Assuming the LHV-LHS model \eqref{lhvextreme} for probability distributions, the correlation terms of our inequality have an LHV-LHS model (LHV model for $\langle A_0\rangle$) if and only if they can be written as:
\begin{equation}
\label{Alhv}
    \begin{split}
        \langle A_0\rangle &= p^{1|0}-p^{-1|0}=2p^{1|0}-1\\
        &= \sum_{\chi}\mathfrak{p}(\chi)(2p^{1|0}(\chi)-1);
    \end{split}
\end{equation}
\begin{equation}
\label{ABlhv}
    \begin{split}
        \langle A_0B_0\rangle &= P(a=b|x=0,y=0)-P(a=-b|x=0,y=0)\\
        &= \sum_{\chi}\int_{\xi}\mathfrak{p}(\chi,\xi)(2p^{1|0}(\chi)-1)(2p^{1|0}(\xi)-1);
    \end{split}
\end{equation}
\begin{equation}
\label{A'B'lhv}
    \langle A_1B_1\rangle= \sum_{\chi}\int_{\xi}\mathfrak{p}(\chi,\xi)(2p^{1|1}(\chi)-1)(2p^{1|1}(\xi)-1).
\end{equation}
In the above expressions we have expressed the separate terms of our inequality as convex sums of extremal probabilities. From eqn. \eqref{obsB'} we see that $\mu=1/2$ refers to the case when the observables $y=0$ and $y=1$ correspond to orthogonal spin measurements which is required in our case since Bob is making spin measurements along $Z,X$, in this case equation \eqref{Bextreme} can further be simplified as (note that $\xi$ below is different from that used in eqn. \eqref{Bextreme} but the symbol has been kept same):
\begin{equation}
    \label{Breextreme}
    \begin{aligned}
    2p^{1|0}(\xi)-1 &= \cos\xi,\\ 
    2p^{1|1}(\xi)-1 &=  \sin\xi.
    \end{aligned}
\end{equation}
Using equations \eqref{Aextreme} and \eqref{Breextreme}, we can further simplify the terms of our inequality for each value of $(\chi,\xi)$ as:
\begin{equation}
\label{extremetable}
    \begin{aligned}
     &\\
        & \langle A_0\rangle \\
        & \langle A_0B_0\rangle (\xi) \\
        & \langle A_1B_1\rangle (\xi)
    \end{aligned}
    \begin{aligned}
        & \quad \chi=1 \\
        & \quad 1 \\
        & \quad \cos\xi \\
        & \quad \sin\xi
    \end{aligned}
    \begin{aligned}
        & \quad \chi=2 \\
        & \quad 1 \\
        & \quad \cos\xi \\
        & \quad -\sin\xi
    \end{aligned}
    \begin{aligned}
        & \quad \chi=3 \\
        & \quad -1 \\
        & \quad -\cos\xi \\
        & \quad \sin\xi
    \end{aligned}
    \begin{aligned}
        & \quad \chi=4 \\
        & \quad -1 \\
        & \quad -\cos\xi \\
        & \quad -\sin\xi
    \end{aligned}
\end{equation}
For $\chi=1$ column in \eqref{extremetable}, the following holds:
\begin{equation}
    \begin{split}
        I_{\alpha,\beta} &\equiv \alpha\langle A_0\rangle+\beta\langle A_0B_0\rangle+\langle A_1B_1\rangle \\
        & =\alpha+\beta \cos\xi+\sin\xi \\
        & \leq \alpha+\sqrt{1+\beta^2}
    \end{split}
\end{equation}
One can similarly verify that for all columns in \eqref{extremetable}, $I_{\alpha,\beta} \leq \alpha+\sqrt{1+\beta^2}$ is satisfied and therefore it must also be satisfied for any convex combination taken over $\chi$ and $\xi$ as shown in equations \eqref{Alhv}, \eqref{ABlhv} and \eqref{A'B'lhv}. By obtaining this local bound, we can say that if for a quantum state $I_{\alpha,\beta}>\alpha+\sqrt{1+\beta^2}$ then the probability statistics obtained from that state do not admit a LHV-LHS decomposition and hence the state is steerable. As a result:
\begin{equation}
\label{local_bound_result}
    I_{\alpha,\beta}\equiv \alpha\langle A_0\rangle+\beta\langle A_0Z\rangle+\langle A_1X\rangle\leq\alpha+\sqrt{1+\beta^2}
\end{equation}
is a valid steering inequality.

\subsection{Quantum bound of Tilted Steering Inequality}
\label{quantum_bound_TSI}
First let's look at the steerable weight of the set of assemblages that achieve the maximal violation of any general steering inequality. One of the interesting properties of steerable weight ($SW(\sigma_{a|x})$) is that it is bounded by any convex function $f(.)$ of that assemblage as given in \cite{cavalcanti_review}, where $f(.)$ can be the violation of a steering inequality. Using this property, it can be further shown that:
\begin{equation}
    SW(\sigma_{a|x}) \geq \frac{f(\sigma_{a|x})-f^{LHS}_{max}}{f_{max}-f^{LHS}_{max}},
\end{equation}
where $f_{max}$ and $f^{LHS}_{max}$ are the maximal value of $f(.)$ among all possible assemblages and among all possible LHS assemblages respectively. From the above equation it can be inferred that SW of an assemblage that gives the maximal violation of a steering inequality must be greater than or equal to 1, which is in fact the maximum value that SW can take. Hence, SW of the assemblage that achieves the maximal violation must be 1.

Now we will identify the candidate states which could achieve the maximal violation. Maximal violation of a steering inequality (maximally steerable states) implies that the initial state that generated the steerable assemblage must be a steerable state and since steerable states are a subset of entangled states, the maximally steerable state is an entangled state. It was proven in \cite{goswami2018one} that steerable weight of an assemblage generated by an arbitrary bipartite qubit mixed entangled state cannot be equal to 1. Therefore, in order to find the maximal quantum violation we focus only on pure 2-qubit entangled states.

A general pure 2-qubit entangled state can be written as (using Schmidt decomposition):
\begin{equation}
\label{purestate}
    |\psi(\theta)\rangle=\cos\theta|00\rangle+\sin\theta|11\rangle; \quad 0<\theta<\pi/2
\end{equation}
Since we have taken the Schmidt decomposition as the generalised pure state, the expression that we are going to maximise would be
\begin{equation}
\label{I_general}
\alpha\langle A_0\rangle + \beta\langle A_0 B_0\rangle + \langle A_1 B_1\rangle
\end{equation}
since optimal Bob's measurements might not be $Z$ and $X$ in the schmidt basis, however $B_0$ and $B_1$ projective measurement directions must remain perpendicular. This maximisation of \eqref{I_general} is equivalent to the maximisation of the TSI since the local bases of Bob can always be rotated to make the optimal Bob's measurements as $Z$ and $X$, however the form of our general entangled state $|\psi(\theta)\rangle$ would change in that case. We proceed to find the quantum bound by maximising \eqref{I_general} over all general pure 2-qubit entangled states and general projective measurements on Alice's and Bob's side which can be written as $A_u=\vec{a_u}.\vec{\sigma}$ and $B_u=\vec{b_u}.\vec{\sigma}$ where $\vec{a_u}=(a_{ux},a_{uy},a_{uz})$ and $\vec{b_u}=(b_{ux},b_{uy},b_{uz})$ are unit vectors with $\vec{b_0} \perp \vec{b_1}$. One can easily verify that the state $|00\rangle$ achieves the value $\alpha + \sqrt{1+\beta^2}$ with measurements $\vec{a_0}=\vec{a_1}=\hat{z}$, $\vec{b_0}=(cos\mu, 0, sin\mu)$ and $\vec{b_1}=(-sin\mu, 0, cos\mu)$ where $cos\mu = 1/\sqrt{1+\beta^2}$, thus the quantum bound is always greater than or equal to the local bound for any choice of parameters $\alpha, \beta$.
Let us now rewrite the state $|\psi(\theta)\rangle$ as $\rho = |\psi(\theta)\rangle \langle \psi(\theta)|$ where,
\begin{equation}
    \rho = \frac{I}{4} + cos2\theta\frac{\sigma_z \otimes I}{4} + cos2\theta\frac{I \otimes \sigma_z}{4} + \sum_{i,j}{T_{ij}\frac{\sigma_i \otimes \sigma_j}{4}},
\end{equation}
and $T_{xx}=sin2\theta$, $T_{yy}=-sin2\theta$, $T_{zz}=1$ and $T_{ij}=0$ for $i\neq j$.
Consider the expression $I_{\beta} \equiv \beta\langle A_0 B_0\rangle + \langle A_1 B_1\rangle$ which is evaluated using $\rho$ and measurements $A_u$ and $B_u$ as:
\begin{equation}
    I_\beta = \beta(\vec{a_0}.T\vec{b_0}) + \vec{a_1}.T\vec{b_1}
\end{equation}
The expression above is maximized over all measurements $A_{0,1}$ and $B_{0,1}$ when $\vec{a_0} \parallel T\vec{b_0}$ and $\vec{a_1} \parallel T\vec{b_1}$ giving:
\begin{equation}
    I_\beta = \beta|T\vec{b_0}| + |T\vec{b_1}|
\end{equation}
with $\vec{b_0} = (cos\mu, 0, sin\mu)$ and $\vec{b_1} =  (-sin\mu, 0, cos\mu)$ since $\vec{b_0}$ and $\vec{b_1}$ can always be considered to be in the X-Z plane. Thus it follows,
\begin{equation}
    I_\beta = \beta\sqrt{sin^2 2\theta cos^2\mu + sin^2\mu} + \sqrt{sin^2 2\theta sin^2\mu + cos^2\mu}
\end{equation}
Differentiating the above expression with respect to $\mu$, we see maxima occurs at $cos^2\mu = \frac{1-\beta^2sin^2 2\theta}{cos^2 2\theta(1+\beta^2)}$ giving $I_\beta^{max} = \sqrt{(1+\beta^2)(1+sin^2 2\theta)}$.
Moving on to the expression \eqref{I_general}, let's first make an observation about the single party correlation $\langle A_0 \rangle$ i.e. for the state $\rho$ and general Pauli observable $A_0$ :  $-cos2\theta \leq \langle A_0 \rangle \leq cos2\theta$ and the extremal values are achieved when $A_0=\pm\sigma_z$. Thus we have,
\begin{equation}
\label{eq_maximisation}
\begin{aligned}
    \alpha\langle A_0\rangle + \beta\langle A_0 B_0\rangle + \langle A_1 B_1\rangle &\equiv \alpha\langle A_0\rangle + I_\beta \\
    & \leq \alpha\, cos2\theta + I_\beta^{max} \\
    & = \alpha\,cos2\theta + \sqrt{(1+\beta^2)(1+sin^2 2\theta)}
\end{aligned}    
\end{equation}
On maximising the expression $\alpha cos2\theta + \sqrt{(1+\beta^2)(1+sin^2 2\theta)}$ w.r.t $\theta$ we see that maxima occurs at $sin^2 2\theta = \frac{1+\beta^2 - \alpha^2}{1+\beta^2 + \alpha^2}$ and the minima at $sin2\theta=0$. Substituting the expression for $sin2\theta$ corresponding to the maxima, we obtain the quantum bound of \eqref{I_general} as $\sqrt{2(1+\beta^2 + \alpha^2)}$. Note that the equality in the second line of \eqref{eq_maximisation} exists only when certain conditions are fulfilled, which are:
\begin{equation}
\label{I_conditions}
    \vec{a_0} \parallel T\vec{b_0}, \, \vec{a_1} \parallel T\vec{b_1} \mbox{ and } A_0 = \sigma_z
\end{equation}
For these conditions to be true we require the $I_\beta^{max}$ to occur at $\vec{b_1} = \hat{x}$ and $\vec{b_0} = \hat{z} \implies cos\mu = 0$ giving $\sin2\theta = \frac{1}{\beta}$. Since the maxima $\sqrt{2(1+\beta^2+\alpha^2)}$ occurs at $sin^2 2\theta = \frac{1+\beta^2 - \alpha^2}{1+\beta^2 + \alpha^2}$, we have:
\begin{equation}
    \frac{1+\beta^2 - \alpha^2}{1+\beta^2 + \alpha^2} = \frac{1}{\beta^2} \implies \beta^2 = \alpha^2 + 1
\end{equation}
Thus for a general state $|\psi(\theta)\rangle$ \eqref{purestate}, if we design the TSI such that $\beta = 1/sin2\theta$ and $\beta^2 = \alpha^2 + 1$ then the quantum bound $\sqrt{2(1+\beta^2+\alpha^2)} (=2\beta)$ is achieved only when Alice and Bob make the measurements $A_0=Z$, $A_1 = X$, $B_0 = Z$ and $B_1 = X$ in their respective local schmidt bases. 

\subsection{SDI certification of 2-qubit pure entangled state using Tilted Steering Inequality}
\label{ST_2qubit}
\begin{lemma}
\label{lem1_all}
In the steering scenario where Bob, the trusted party, performs mutually unbiased spin measurements along Z and X; maximal violation ($\sqrt{2(1+\beta^2+\alpha^2)}$) of the Tilted Steering Inequality (with $\beta^2 = \alpha^2 + 1$) occurs if and only if the 2-qubit joint state shared between Alice and Bob is:
\begin{equation}
\label{qubit_state}
    |\psi(\theta)\rangle=\cos\theta|00\rangle+\sin\theta|11\rangle; \mbox{ such that\ }\sin2\theta=\frac{1}{\beta}=\frac{1}{\sqrt{1+\alpha^2}}
\end{equation}
up to local unitary transformations and Alice performs spin measurements along Z and X i.e. $A_0=|0\rangle\langle0|-|1\rangle\langle1|$ and $A_1=|+\rangle\langle+|-|-\rangle\langle-|$ (or their local unitary equivalents) where $|\pm\rangle=\frac{|0\rangle\pm|1\rangle}{\sqrt{2}}$. 
\end{lemma}
\begin{proof}
We are given that the joint state generates an assemblage that violates the tilted steering inequality maximally, therefore the joint state is steerable and hence entangled. Using the arguments from the previous section, we can eliminate mixed entangled states implying that the joint state is a pure entangled state. Also, we have shown that the maximal violation occurs for the state $|\psi(\theta)\rangle$ only if Alice and Bob make spin measurements along  $Z$ and $X$ and for given values of $\beta$ and $\alpha$ satisfying $\beta^2 = \alpha^2 + 1$ there is a unique $\theta$ in $(0,\pi/4)$ given by $\sin2\theta=\frac{1}{\beta}$. Further, any pure entangled state can be expressed as $|\psi(\theta)\rangle$ using schmidt decomposition and for this form of the state the only measurements that achieve the maximal violation are the ones shown above; hence we can conclude that no other locally rotated state could achieve the maximal violation with the measurements defined above in the schmidt bases.
\end{proof}

\begin{result}
\label{Result_1_2qubit}
Maximal violation of the Tilted Steering inequality certifies any 2-qubit pure entangled state i.e. for the inequality $\alpha\langle A_0\rangle+\beta\langle A_0Z\rangle+\langle A_1X\rangle\leq\alpha+\sqrt{1+\beta^2}$ with $\beta^2 = \alpha^2 + 1$, let the maximal violation be achieved by an assemblage generated by performing measurements on the joint pure state $|\psi\rangle_{AB}\in\mathcal{H}_A\otimes\mathcal{H}_B$ (where the dimensions of the trusted side B is 2) and measurement operators $\{M_{a|x}\}_{a,x}$ on Alice's side. Then there exists an isometry on Alice's side $\Phi:\mathcal{H}_A\to \mathcal{H}_A\otimes\mathcal{H}_{A'}$, where the dimension of $\mathcal{H}_{A'}$ is 2, such that,
\begin{equation}
    \begin{aligned}
        \Phi(|\psi\rangle_{AB})&=|junk\rangle_A\otimes|\psi(\theta)\rangle_{A'B},\\
        \Phi(M_{a|x}\otimes I|\psi\rangle_{AB})&=|junk\rangle_A\otimes (\tilde{M}_{a|x}\otimes I)|\psi(\theta)\rangle_{A'B}
    \end{aligned}
\end{equation}
where $|\psi(\theta)\rangle_{A'B}$ is given by \eqref{qubit_state} and measurement operators $\tilde{M}_{a|x}$ correspond to observables $A_0=\sigma_z$ and $A_1=\sigma_x$.
\end{result}
\begin{proof}
Similar result for state certification from the maximal violation of fine grained steering inequality \cite{fine_grained_ineq} was proved in \cite{goswami2018one}. According to Jordan's lemma (see \cite{masanes} for proof), if there are operators $A_0$ and $A_1$ with eigenvalues $\pm1$ acting on a Hilbert space $\mathcal{H}$  of dimension $d$, then there is a decomposition of $\mathcal{H}$ as a direct sum of subspaces $\mathcal{H}_i$ with dimension $d\leq2$ each. That means $A_0$ and $A_1$ act within each subspace $\mathcal{H}_i$ and can be written as $A_0=\oplus_i A_0^i$ and $A_1=\oplus_i A_1^i$.\\
In our steering case, operators on the untrusted side (Alice) act on a Hilbert space $\mathcal{H}_A$ with dimension $d$. From Jordan's lemma, we can say that these observables act on subspace $\mathcal{H}_A^i$ with dimension $\leq 2$ i.e.,
\begin{equation}
    \mathcal{H}_A\otimes\mathcal{H}_B=\oplus_i(\mathcal{H}_A^i\otimes\mathcal{H}_B)
\end{equation}
Consider that $A_x=\Pi_{0|A_x}-\Pi_{1|A_x}$ with $x\in \{0,1\}$, where $\Pi_{a|A_x}$ is the projector acting on $\mathcal{H}_A$. Hence, we can write $\Pi_{a|A_x}=\oplus_i \Pi^{i}_{a|A_x}$ such that each $\Pi^{i}_{a|A_x}$ acts on $\mathcal{H}_A^{i}$ for all $a$ and $x$; we also denote as $\Pi^i = \Pi^{i}_{0|A_x} + \Pi^{i}_{1|A_x}$ the projector on $\mathcal{H}_A^{i}$. Also, we denote the projector associated with outcome $b$ of measurement $y$ as $\Pi_{b|B_y}$ acting on $\mathcal{H}_B$ of dimension 2.\\
Now, for any state $\rho_{AB}\in \mathcal{B}(\mathcal{H}_A \otimes \mathcal{H}_B)$ we have;
\begin{equation}
\label{subspace_prob}
    \begin{aligned}
        P(a,b|x,y) &= Tr[\rho_{AB}(\Pi_{a|A_x} \otimes \Pi_{b|B_y})] \\
        &= \sum_{i}{q_i Tr[\rho_{AB}^{i}(\Pi^{i}_{a|A_x} \otimes \Pi_{b|B_y})]} \\
        &= \sum_i{q_i P_i(a,b|x,y)}
    \end{aligned}
\end{equation}
where $q_i=Tr[\rho_{AB}(\Pi^{i} \otimes I)] \geq 0 \, \forall \, i; \sum_i{q_i}=1$,  $\rho^i_{AB} = \frac{(\Pi^i\otimes I)\rho_{AB}(\Pi^i\otimes I)}{q_i} \in \mathcal{B}(\mathcal{H}_A \otimes \mathcal{H}_B)$ is, at most, a 2-qubit state and $P_i(a,b|x,y)$ is the joint probability in the 2-qubit subspace of Alice. \\
Using \eqref{subspace_prob}, we can write the following correlation terms as:
\begin{equation}
    \begin{aligned}
        \langle A_0 Z\rangle &= \sum_i{q_i\langle A_0 Z\rangle_i} \\
        \langle A_1 X\rangle &= \sum_i{q_i\langle A_1 X\rangle_i} \\
        \langle A_0 \rangle &= \sum_i{q_i\langle A_0 \rangle_i}
    \end{aligned}
\end{equation}
and as a consequence we can write the Tilted Steering expression as:
\begin{equation}
\label{TSI_subspace}
\begin{aligned}
    I_{\alpha,\beta} &\equiv \alpha\langle A_0\rangle + \beta\langle A_0Z\rangle + \langle A_1X\rangle \\
    &= \sum_{i}{q_i(\alpha\langle A_0\rangle_i + \beta\langle A_0Z\rangle_i + \langle A_1X\rangle_i)} \\
    &= \sum_{i}{q_i I_{\alpha,\beta}^i}
\end{aligned}
\end{equation}
Note that for the maximum value of $I_{\alpha,\beta}$ to be obtained it is necessary that, for all $i$ such that $q_i\neq 0$, the dimension of each $\mathcal{H}_A^i$ be equal to 2. This implies that the effective dimension $d$ of the local Hilbert space of Alice $\mathcal{H}_A$ is even. Further, from \eqref{TSI_subspace} it follows that $I_{\alpha,\beta}$ is a convex sum over $I_{\alpha,\beta}^i$ and the maximal violation of $I_{\alpha,\beta}$ is obtained if and only if $I_{\alpha,\beta}^{i}$ is maximal ($=\sqrt{2(1+\beta^2+\alpha^2)}$). From \autoref{lem1_all}, the maximal quantum value of $I_{\alpha,\beta}^i$ is achieved if and only if the shared 2-qubit state $\rho_{AB}^i$ is a pure entangled state, i.e., $\rho_{AB}^i=|\psi^i(\theta)\rangle \langle\psi^i(\theta)|$ where:
\begin{equation}
\label{subspace_state}
    |\psi^i(\theta)\rangle = cos\theta|2i,0\rangle + sin\theta |2i+1,1\rangle
\end{equation}
with $sin2\theta=\frac{1}{\beta}=\frac{1}{\sqrt{1+\alpha^2}}$ and projectors given as $\Pi^i_{0|A_0} = |2i\rangle \langle2i|$, $\Pi^i_{1|A_0} = |2i+1\rangle \langle2i+1|$, $\Pi^i_{0|A_1} = |+_i\rangle \langle+_i|$ and $\Pi^i_{1|A_1} = |-_i\rangle \langle-_i|$; $|\pm_i\rangle = \frac{|2i\rangle \pm |2i+1\rangle}{\sqrt{2}}$. Here, the state parameter $\theta$ is uniquely determined by $\beta$ (or $\alpha$) specified by the particular TSI being used for the certification task. \\
Hence, the state $|\psi\rangle_{AB} \in \mathcal{H}_A\otimes\mathcal{H}_B$ and measurement projection operators $\{M_{a|x}\}_{a,x}$ give the maximal violation if and only if they can be decomposed as:
\begin{equation}
    |\psi\rangle_{AB}=\oplus_i\sqrt{q_i}|\psi^i(\theta)\rangle; \quad M_{a|x}=\oplus_i \Pi_{a|x}^i
\end{equation} 
where $\sum_i q_i=1$ and $|\psi^i(\theta)\rangle$ is given by \eqref{subspace_state}.

Then one can append an ancilla qubit on Alice's side prepared in state $|0\rangle_{A'}$ and define a local isometry as the map:
\begin{equation}
\begin{aligned}
    \Phi|\oplus_i 2i,0\rangle_{AA'}&\longmapsto|\oplus_i 2i,0\rangle_{AA'},\\
    \Phi|\oplus_i 2i+1,0\rangle_{AA'}&\longmapsto|\oplus_i 2i,1\rangle_{AA'}.
\end{aligned}
\end{equation}
which can be used as:
\begin{equation}
    (\Phi\otimes I)|\psi\rangle_{AB}|0\rangle_{A'} = |junk\rangle_A\otimes|\psi(\theta)\rangle_{A'B}
\end{equation}
to extract the target state $|\psi(\theta)\rangle_{A'B}$ \eqref{qubit_state} from $|\psi\rangle_{AB}$.
\end{proof}

\section{Proof of SDI certification for arbitrary bipartite pure entangled states}
\label{SDI_proof_ST}
Let us first describe the structure of the ideal assemblage that we would look for on Bob's side, then we will proceed to show how such a structure of the assemblage can be used for certifying the target state \eqref{qudit_state}. 
\subsection{Structure of the certifying assemblage}
\label{assemblage_structure}
In order to obtain the complete set of ideal assemblages that meet the sufficient condition for SDI certification, we need 3 measurement settings on Alice's side with $d$ outcomes each. The assemblage set generated from Alice's measurement settings $x\in \{0,1\}$ will be used to certify subspace projections of the form $c_{2m}|2m,2m\rangle+c_{2m+1}|2m+1,2m+1\rangle$ and those generated from measurement settings $x\in \{0,2\}$ will certify projections of the form $c_{2m+1}|2m+1,2m+1\rangle+c_{2m+2}|2m+2,2m+2\rangle$ of the target state.\\

\noindent {\em Assemblage obtained from measurement $x=0$: } 
For this case, we impose a condition that the assemblages corresponding to different outcomes only have diagonal non-zero terms;
\begin{equation}
    \sigma_{i|0}=c_i^2|i\rangle \langle i|\, \forall\, i\in\{0,1,\ldots,d-1 \}
\end{equation}\\

\noindent {\em Assemblage obtained from measurement $x=1$: } 
In this case, we require the assemblage to have certain non-zero non-diagonal terms along with diagonal terms. Specifically, the assemblages should have non-zero outer product between eigenvectors $|2m\rangle$ and $|2m+1\rangle$ in the computational basis;
\begin{table}[h!]
\caption{Non-zero elements of $\sigma_{2m|1}\mbox{ and }\sigma_{2m+1|1}\, \forall\, m\in\{0,\ldots,\frac{d}{2}-1\}$ for even $d\geq 2$}
\label{table1}
 \begin{tabular}{|c||c|c|} 
 \hline
 $\sigma_{2m|1}$ & $\langle2m|$ & $\langle 2m+1|$ \\ [0.5ex] 
 \hline\hline
 $|2m\rangle$ & $\frac{c_{2m}^2}{2}$ & $\frac{c_{2m}c_{2m+1}}{2}$  \\ 
 \hline
 $|2m+1\rangle$ & $\frac{c_{2m}c_{2m+1}}{2}$  & $\frac{c_{2m+1}^2}{2}$ \\ [1ex] 
 \hline
\end{tabular}
\quad\quad
 \begin{tabular}{|c||c|c|} 
 \hline
 $\sigma_{2m+1|1}$ & $\langle2m|$ & $\langle 2m+1|$ \\ [0.5ex] 
 \hline\hline
 $|2m\rangle$ & $\frac{c_{2m}^2}{2}$ & $-\frac{c_{2m}c_{2m+1}}{2}$  \\ 
 \hline
 $|2m+1\rangle$ & $-\frac{c_{2m}c_{2m+1}}{2}$  & $\frac{c_{2m+1}^2}{2}$ \\ [1ex]
 \hline
\end{tabular}
\end{table}

\begin{table}[h!]
\caption{Non-zero elements of $\sigma_{2m|1}\mbox{ and }\sigma_{2m+1|1}\, \forall\, m\in\{0,\ldots,\frac{d-1}{2}-1\}$ for odd $d\geq 3$}
\label{table2}
 \begin{tabular}{|c||c|c|} 
 \hline
 $\sigma_{2m|1}$ & $\langle2m|$ & $\langle 2m+1|$ \\ [0.5ex] 
 \hline\hline
 $|2m\rangle$ & $\frac{c_{2m}^2}{2}$ & $\frac{c_{2m}c_{2m+1}}{2}$  \\ 
 \hline
 $|2m+1\rangle$ & $\frac{c_{2m}c_{2m+1}}{2}$  & $\frac{c_{2m+1}^2}{2}$ \\ [1ex] 
 \hline
\end{tabular}
\quad
 \begin{tabular}{|c||c|c|} 
 \hline
 $\sigma_{2m+1|1}$ & $\langle2m|$ & $\langle 2m+1|$ \\ [0.5ex] 
 \hline\hline
 $|2m\rangle$ & $\frac{c_{2m}^2}{2}$ & $-\frac{c_{2m}c_{2m+1}}{2}$  \\ 
 \hline
 $|2m+1\rangle$ & $-\frac{c_{2m}c_{2m+1}}{2}$  & $\frac{c_{2m+1}^2}{2}$ \\ [1ex]  
 \hline
\end{tabular}
\quad
$\sigma_{d-1|1}=c_{d-1}^2|d-1\rangle\langle d-1|$
\end{table}

\noindent {\em Assemblage obtained from measurement $x=2$: } 
As in the previous case, this measurement will also give rise to an assemblage with non-zero non-diagonal elements but here they will be associated to the outer products between eigenvectors $|2m+1 \mbox{ mod }d\rangle$ and $|2m+2 \mbox{ mod }d\rangle$ in the computational basis;

\begin{table}[h!]
\caption{Non-zero elements of $\sigma_{2m+1|2}\mbox{ and }\sigma_{2m+2|2}\, \forall\, m\in\{0,\ldots,\frac{d}{2}-1\}$ for even $d\geq 2$}
\label{table3}
 \begin{tabular}{|c||c|c|} 
 \hline
 $\sigma_{2m+1|2}$ & $\langle2m+1|$ & $\langle 2m+2|$ \\ [0.5ex] 
 \hline\hline
 $|2m+1\rangle$ & $\frac{c_{2m+1}^2}{2}$ & $\frac{c_{2m+1}c_{2m+2}}{2}$  \\ 
 \hline
 $|2m+2\rangle$ & $\frac{c_{2m+1}c_{2m+2}}{2}$  & $\frac{c_{2m+2}^2}{2}$ \\ [1ex] 
 \hline
\end{tabular}
\quad\quad
 \begin{tabular}{|c||c|c|} 
 \hline
 $\sigma_{2m+2|2}$ & $\langle2m+1|$ & $\langle 2m+2|$ \\ [0.5ex] 
 \hline\hline
 $|2m+1\rangle$ & $\frac{c_{2m+1}^2}{2}$ & $-\frac{c_{2m+1}c_{2m+2}}{2}$  \\ 
 \hline
 $|2m+2\rangle$ & $-\frac{c_{2m+1}c_{2m+2}}{2}$  & $\frac{c_{2m+2}^2}{2}$ \\ [1ex] 
 \hline
\end{tabular}
\end{table}

\begin{table}[h!]
\caption{Non-zero elements of $\sigma_{2m+1|2}\mbox{ and }\sigma_{2m+2|2}\, \forall\, m\in\{0,\ldots,\frac{d-1}{2}-1\}$ for odd $d\geq 3$}
\label{table4}
$\sigma_{0|2}=c_{0}^2|0\rangle\langle 0|$
\quad
 \begin{tabular}{|c||c|c|} 
 \hline
 $\sigma_{2m+1|2}$ & $\langle2m+1|$ & $\langle 2m+2|$ \\ [0.5ex] 
 \hline\hline
 $|2m+1\rangle$ & $\frac{c_{2m+1}^2}{2}$ & $\frac{c_{2m+1}c_{2m+2}}{2}$  \\ 
 \hline
 $|2m+2\rangle$ & $\frac{c_{2m+1}c_{2m+2}}{2}$  & $\frac{c_{2m+2}^2}{2}$ \\ [1ex] 
 \hline
\end{tabular}
\quad
 \begin{tabular}{|c||c|c|} 
 \hline
 $\sigma_{2m+2|2}$ & $\langle2m+1|$ & $\langle 2m+2|$ \\ [0.5ex] 
 \hline\hline
 $|2m+1\rangle$ & $\frac{c_{2m+1}^2}{2}$ & $-\frac{c_{2m+1}c_{2m+2}}{2}$  \\ 
 \hline
 $|2m+2\rangle$ & $-\frac{c_{2m+1}c_{2m+2}}{2}$  & $\frac{c_{2m+2}^2}{2}$ \\ [1ex] 
 \hline
\end{tabular}
\end{table}

\subsection{Ideal measurements on Alice's side}
\label{ideal_measurements}
We will now mention the ideal measurements on Alice's side that give rise to the required assemblages. For $x=0$: measurement is done in the computational basis i.e. in the basis $\{|0\rangle,|1\rangle,\ldots,|d-1\rangle\}$; for $x=1$: for even $d$, measurement is done in the eigenbasis of observable $\oplus_{m=0}^{\frac{d}{2}-1}[\sigma_x]_m$ and for odd $d$, measurement is done in the eigenbasis of the observable $\oplus_{m=0}^{\frac{d-1}{2}-1}[\sigma_x]_m\oplus|d-1\rangle\langle d-1|$ where $[\sigma_x]_m=|2m\rangle\langle2m+1|+|2m+1\rangle\langle2m|$ is defined with respect to the basis $\{|2m \mbox{ mod }d\rangle,|2m+1 \mbox{ mod }d\rangle\}$; finally for $x=2$: for even $d$, measurement is done in the eigenbasis of observable $\oplus_{m=0}^{\frac{d}{2}-1}[\sigma_x]_m'$ and for odd $d$, measurement is done in the eigenbasis of the observable $|0\rangle\langle 0|\oplus_{m=0}^{\frac{d-1}{2}-1}[\sigma_x]_m'$ where $[\sigma_x]_m'=|2m+1\rangle\langle2m+2|+|2m+2\rangle\langle2m+1|$ is defined with respect to the basis $\{|2m+1 \mbox{ mod }d\rangle,|2m+2 \mbox{ mod }d\rangle\}$.

\subsection{Obtaining the sufficient conditions for SDI certification}
\label{unitaries_def}
Analogous to the DI case \cite{yang_navascues, coladangelo2017all}, we will show that the ideal assemblage structure can be used to obtain the condition given in \eqref{conditions} which is sufficient to show the existence of an isometry (Fig. \ref{Alice_isometry}) that certifies the target state. For this purpose we will construct projections $P_A^{(k)}$ and the unitaries $X_A^{(k)}$ using the measurement projection operators on Alice's side that give rise to the ideal assemblages. Let us denote Alice's projection operators for measurement $x$ and outcome $i$ by $\Pi_i^{A_x}$ and explore the consequences of imposing the ideal structure on the generated assemblage as shown in \ref{assemblage_structure}.\\
Define the operator $A_{x,m}=\Pi_{2m}^{A_x}-\Pi_{2m+1}^{A_x}$ for $x\in\{0,1\}$; it can be seen that $(A_{x,m})^2=\Pi_{2m}^{A_x}+\Pi_{2m+1}^{A_x}:=\mathds{1}_m^{A_x}$. Further note that $\norm{\Pi_{i}^{A_0}|\psi\rangle}=\sqrt{\langle\psi|\Pi_{i}^{A_0}|\psi\rangle}=\sqrt{tr(\sigma_{i|A_0})}=c_i$, similarly one can calculate $\norm{\mathds{1}_m^{A_x}|\psi\rangle}=\sqrt{c_{2m}^2+c_{2m+1}^2}\quad\forall x\in\{0,1\}$. Also, define Pauli X, Z and identity operators on Bob's side for the basis $|2m \mbox{ mod } d\rangle, |2m+1 \mbox{ mod } d\rangle$ as $Z_{B,m}=|2m\rangle\langle2m|-|2m+1\rangle\langle2m+1|$, $X_{B,m}=|2m\rangle\langle2m+1|+|2m+1\rangle\langle2m|$ and $\mathds{1}_{B,m}=|2m\rangle\langle2m|+|2m+1\rangle\langle2m+1|$.\\
Due to the structure imposed on the assemblages, we see that:
\begin{equation}
    \begin{aligned}
        I_{\alpha_m,\beta_m}&\equiv \alpha_m\langle\psi|A_{0,m}\mathds{1}_{B,m}|\psi\rangle+\beta_m\langle\psi|A_{0,m}Z_{B,m}|\psi\rangle+\langle\psi| A_{1,m}X_{B,m}|\psi\rangle, \mbox{ with } \beta_m^2=\alpha_m^2+1\\
        &=\alpha_m tr(\sigma_{2m|0}-\sigma_{2m+1|0})+\beta_m tr((\sigma_{2m|0}-\sigma_{2m+1|0})Z_{B,m})+tr((\sigma_{2m|1}-\sigma_{2m+1|1})X_{B,m})\\
        &=\alpha_m (c_{2m}^2-c_{2m+1}^2) + \beta_m(c_{2m}^2+c_{2m+1}^2)+2c_{2m}c_{2m+1}\\
        &=(\beta_m+\sqrt{1+\alpha_m^2})(c_{2m}^2+c_{2m+1}^2)\\
        &=2\beta_m(c_{2m}^2+c_{2m+1}^2); \mbox{ note that } \sqrt{2(1+\alpha_m^2+\beta_m^2)}=2\beta_m
    \end{aligned}
\end{equation}
where the state parameter $\theta_m$ is such that $\sin2\theta_m=\frac{1}{\beta_m}=\frac{2c_{2m}c_{2m+1}}{c_{2m}^2+c_{2m+1}^2}$. Although this is not the maximal violation of the Tilted Steering Inequality but we can use a trick by defining the normalised state $|\psi_m\rangle=\frac{\mathds{1}_m^{A_0}|\psi\rangle}{\sqrt{c_{2m}^2+c_{2m+1}^2}}$; consequently the assemblage modifies as $\sigma_{i|x,m}=\frac{\sigma_{i|x}}{c_{2m}^2+c_{2m+1}^2}$. As a result, the tilted steering expression evaluates to:
\begin{equation}
\label{max_violation_m}
    \begin{aligned}
        I_{\alpha_m,\beta_m}&\equiv \alpha_m\langle\psi_m|A_{0,m}\mathds{1}_{B,m}|\psi_m\rangle+\beta_m\langle\psi_m|A_{0,m}Z_{B,m}|\psi_m\rangle+\langle\psi_m| A_{1,m}X_{B,m}|\psi_m\rangle\\
        &=\alpha_m tr(\sigma_{2m|0,m}-\sigma_{2m+1|0,m})+\beta_m tr((\sigma_{2m|0,m}-\sigma_{2m+1|0,m})Z_{B,m})+tr((\sigma_{2m|1,m}-\sigma_{2m+1|1,m})X_{B,m})\\
        &=\beta_m+\sqrt{1+\alpha_m^2}=2\beta_m
    \end{aligned}
\end{equation}
which is the maximal violation of the inequality given by $I_{\alpha_m,\beta_m}$. Then, by \autoref{lem1_all}, we can deduce that:
\begin{equation}
\label{maximal_state}
    |\psi_m\rangle=\frac{c_{2m}|2m,2m\rangle+c_{2m+1}|2m+1,2m+1\rangle}{\sqrt{c_{2m}^2+c_{2m+1}^2}} \implies \mathds{1}_m^{A_0}|\psi\rangle=c_{2m}|2m,2m\rangle+c_{2m+1}|2m+1,2m+1\rangle
\end{equation}
\begin{equation}
\begin{aligned}
\label{maximal_measurements}
    A_{0,m}&=|2m\rangle\langle2m|-|2m+1\rangle\langle2m+1| \\
    A_{1,m}&=|2m\rangle\langle2m+1|+|2m+1\rangle\langle2m|
    \end{aligned}
\end{equation}

Now that we have the maximal violation, we need to define the \textit{unitarized} versions of our operators as $A_{x,m}^u:=\mathds{1}-\mathds{1}_m^{A_x}+A_{x,m}$ and let $X_{A,m}^u:=A_{1,m}^u$, this unitarization is necessary to achieve the conditions of \autoref{sufficient_lemma}. As can be checked, the maximal violation in \eqref{max_violation_m} also holds with unitarized operators. 

Next let us define the projections $P_A^{(2m)}:=(\mathds{1}_m^{A_0}+A_{0,m})/2=\Pi_{2m}^{A_0}$ and $P_A^{(2m+1)}:=(\mathds{1}_m^{A_0}-A_{0,m})/2=\Pi_{2m+1}^{A_0}$. So, for all $m$ and $k=2m,2m+1$:
\begin{equation}
\begin{aligned}
    P_A^{(k)}|\psi\rangle&=(\mathds{1}_m^{A_0}+(-1)^k A_{0,m})/2|\psi\rangle \\
    &= (\mathds{1}_m^{A_0}|\psi\rangle+(-1)^k A_{0,m}|\psi\rangle)/2 \\
    &= (\mathds{1}_m^{A_0}|\psi\rangle+(-1)^k A_{0,m}\mathds{1}_m^{A_0}|\psi\rangle)/2 \mbox{ since } A_{0,m}=A_{0,m}\mathds{1}_m^{A_0} \\
    &= (c_{2m}|2m,2m\rangle+c_{2m+1}|2m+1,2m+1\rangle +(-1)^k(c_{2m}|2m,2m\rangle-c_{2m+1}|2m+1,2m+1\rangle))/2\\
    &= c_k|k,k\rangle
\end{aligned}
\end{equation}
Further, see that:
\begin{equation}
\begin{aligned}
    X_{A,m}^u P_A^{(2m+1)}|\psi\rangle &= X_{A,m}^u c_{2m+1}|2m+1,2m+1\rangle \\
    &= A_{1,m}c_{2m+1}|2m+1,2m+1\rangle \\
    &= c_{2m+1}|2m,2m+1\rangle
\end{aligned}
\end{equation}

Similarly, we can perform analogous calculations for the operators: $A'_{0,m}=\Pi_{2m+1}^{A_0}-\Pi_{2m+2}^{A_0}$ and $A'_{1,m}=\Pi_{2m+1}^{A_2}-\Pi_{2m+2}^{A_2}$ for $x\in\{0,2\}$. Following the same procedure as shown above, we define the unitary operators $A'^{u}_{x,m}:=\mathds{1}-\mathds{1}_m^{A'_x}+A'_{x,m}$ and let $Y_{A,m}^u:=A'^u_{1,m}$; further we can obtain:
\begin{gather*}
    P_A^{(k)}|\psi\rangle= c_k|k,k\rangle \mbox{ for all }\, k=2m+1,2m+2 \\
    Y_{A,m}^u = \mathds{1}-\mathds{1}_m^{A'_1}+(|2m+1\rangle\langle2m+2|+|2m+2\rangle\langle2m+1|) \\
    Y_{A,m}^u P_A^{(2m+2)}|\psi\rangle = c_{2m+2}|2m+1,2m+2\rangle
\end{gather*}

As the final step we construct the unitaries $X_A^{(k)}$ required in \autoref{sufficient_lemma} from $X_{A,m}^u$ and $Y_{A,m}^u$ as:
\begin{equation}
\label{big_unitary}
    X_A^{(k)}=
    \begin{cases}
    \mathds{1}, \mbox{ for } k=0 \\
    X_{A,0}^u Y_{A,0}^u X_{A,1}^u Y_{A,1}^u \ldots X_{A,m-1}^u Y_{A,m-1}^u X_{A,m}^u, \mbox{ for } k=2m+1 \\
    X_{A,0}^u Y_{A,0}^u X_{A,1}^u Y_{A,1}^u \ldots X_{A,m-1}^u Y_{A,m-1}^u, \mbox{ for } k=2m \\
    \end{cases}
\end{equation}
Note that $X_A^{(k)}$ is unitary since it is composed of unitaries. Let us now verify the condition required in \autoref{sufficient_lemma}, i.e.;
\begin{equation}
\label{sufficient_condition}
    X_A^{(k)}P_A^{(k)}|\psi\rangle=c_k|0,k\rangle
\end{equation}
for different cases.\\
For case $k=0$:
\begin{equation}
\begin{aligned}
    X_A^{(0)}P_A^{(0)}|\psi\rangle &= \mathds{1}P_A^{(0)}|\psi\rangle \\
    &= c_0|0,0\rangle
\end{aligned}
\end{equation}
For case $k=2m+1$:
\begin{equation}
\begin{aligned}
    X_A^{(2m+1)}P_A^{(2m+1)}|\psi\rangle &= X_{A,0}^u Y_{A,0}^u X_{A,1}^u Y_{A,1}^u \ldots X_{A,m-1}^u Y_{A,m-1}^u X_{A,m}^u P_A^{(2m+1)}|\psi\rangle \\
    &= X_{A,0}^u Y_{A,0}^u X_{A,1}^u Y_{A,1}^u \ldots X_{A,m-1}^u Y_{A,m-1}^u c_{2m+1}|2m,2m+1\rangle \\
    &= X_{A,0}^u Y_{A,0}^u X_{A,1}^u Y_{A,1}^u \ldots X_{A,m-1}^u c_{2m+1}|2m-1,2m+1\rangle\\
    &= X_{A,0}^u Y_{A,0}^u X_{A,1}^u Y_{A,1}^u \ldots c_{2m+1}|2m-2,2m+1\rangle \\
    &= \ldots \\
    &= c_{2m+1}|0,2m+1\rangle
\end{aligned}
\end{equation}
For the final case $k=2m$:
\begin{equation}
\begin{aligned}
    X_A^{(2m)}P_A^{(2m)}|\psi\rangle &= X_{A,0}^u Y_{A,0}^u X_{A,1}^u Y_{A,1}^u \ldots X_{A,m-1}^u Y_{A,m-1}^u P_A^{(2m)}|\psi\rangle \\
    &= X_{A,0}^u Y_{A,0}^u X_{A,1}^u Y_{A,1}^u \ldots X_{A,m-1}^u Y_{A,m-1}^u c_{2m}|2m,2m\rangle \\
    &= \ldots \\
    &= c_{2m}|0,2m\rangle
\end{aligned}
\end{equation}

\subsection{From the sufficient conditions to SDI state certification}
\label{conditions}
In the previous subsection we verified that the condition \eqref{sufficient_condition} holds, in the following lemma we claim that the condition is sufficient for showing the existence of an isometry on Alice's side which can be used to extract out the target state \eqref{qudit_state} from the joint state between Alice and Bob.
\begin{lemma}
\label{sufficient_lemma}
Suppose there exist unitary operator $X_A^{(k)}$ and projections $\{P_A^{(k)}\}_{k=0,\ldots,d-1}$ which form a complete orthogonal set and they satisfy the following condition:
\begin{equation}
\label{lem_condition}
    X_A^{(k)}P_A^{(k)}|\psi\rangle=c_k|0,k\rangle \quad \forall k
\end{equation}
then there exists a local isometry on Alice's side $\Phi$ such that $\Phi(|\psi\rangle)=|junk\rangle\otimes|\psi_{target}\rangle$, for some auxiliary state $|junk\rangle$.
\end{lemma}

\begin{proof}
We will consider a local isometry on Alice's side which is the $d$-dimensional generalisation of SWAP isometry, introduced in \cite{yang_navascues}:
\begin{equation}
\label{alice_iso}
    \Phi:=R_{AA'}\Bar{F}_{A'}S_{AA'}F_{A'}
\end{equation}
where $F$ is the quantum Fourier transform, $\Bar{F}$ is the inverse quantum Fourier transform, $R_{AA'}$ is defined as $R_{AA'}|\psi\rangle_{AB}|k\rangle_{A'}=X_A^{(k)}|\psi\rangle_{AB}|k\rangle_{A'}$ and $S_{AA'}$ is defined as $S_{AA'}|\psi\rangle_{AB}|k\rangle_{A'}=Z_A^{k}|\psi\rangle_{AB}|k\rangle_{A'}$ where $Z_A:=\sum_{k=0}^{d-1} \omega^k P_A^{(k)}$. Let us now calculate the consequence of applying this isometry:
\begin{align}
    \Phi|\psi\rangle_{AB}|0\rangle_{A'} &= R_{AA'}\Bar{F}_{A'}S_{AA'}F_{A'} |\psi\rangle_{AB}|0\rangle_{A'} \\
    &= R_{AA'}\Bar{F}_{A'}S_{AA'} \frac{1}{\sqrt{d}} \sum_k |\psi\rangle_{AB}|k\rangle_{A'} \\
    &= R_{AA'}\Bar{F}_{A'} \frac{1}{\sqrt{d}} \sum_k \left(\sum_{j=0}^{d-1}\omega^j P_A^{(j)}\right)^k|\psi\rangle_{AB}|k\rangle_{A'} \\
    &= R_{AA'}\Bar{F}_{A'} \frac{1}{\sqrt{d}} \sum_{k,j} \omega^{jk} P_A^{(j)} |\psi\rangle_{AB}|k\rangle_{A'} \\
    &= R_{AA'} \frac{1}{d} \sum_{k,j,l} \omega^{jk} \omega^{-lk}  P_A^{(j)} |\psi\rangle_{AB}|l\rangle_{A'} \\
    &= R_{AA'} \sum_j P_A^{(j)} |\psi\rangle_{AB}|j\rangle_{A'} \label{b20} \\
    &= \sum_j X_A^{(j)} P_A^{(j)} |\psi\rangle_{AB}|j\rangle_{A'} \label{b21} \\
    &= \sum_j c_j|0,j\rangle_{AB}|j\rangle_{A'} \mbox{ using the given condition \eqref{lem_condition} } \\
    &= |0\rangle_A \otimes \sum_j c_j|j,j\rangle_{A'B} \\
    &= |junk\rangle_A \otimes |\psi_{target}\rangle_{A'B}
\end{align}
\end{proof}
Having certified the target state, the same isometry $\Phi$ \eqref{alice_iso} can be used to certify Alice's ideal measurements. Using the same notations from previous sections $A_{x,m}=\Pi_{2m}^{A_x}-\Pi_{2m+1}^{A_x}$ for $x\in\{0,1\}$ and let the 2-qubit ideal measurements be denoted by $[\sigma_z]_m$ for $x=0$ and $[\sigma_x]_m$ for $x=1$ on the $(2m,2m+1)$ subspace (see section \ref{ideal_measurements}). Then, given that the observed assemblage has the ideal structure, consider the case $A_{x,m}$ for $x=1$:
\begin{equation}
    \begin{aligned}
    \Phi(A_{1,m}|\psi\rangle_{AB}|0\rangle_{A'}) &= R_{AA'} \sum_j P_A^{(j)} A_{1,m}|\psi\rangle_{AB}|j\rangle_{A'}\quad \mbox{ continuing from \eqref{b20}} \\
    &= R_{AA'} \sum_j P_A^{(j)} A_{1,m}\mathds{1}^{A_0}_m|\psi\rangle_{AB}|j\rangle_{A'} \\
    &= R_{AA'} \sum_j P_A^{(j)} A_{1,m}(c_{2m}|2m,2m\rangle_{AB}+c_{2m+1}|2m+1,2m+1\rangle_{AB})|j\rangle_{A'}\quad \mbox{ from \eqref{maximal_state}} \\
    &= R_{AA'} \sum_j P_A^{(j)}(c_{2m}|2m+1,2m\rangle_{AB}+c_{2m+1}|2m,2m+1\rangle_{AB})|j\rangle_{A'}\quad \mbox{ from \eqref{maximal_measurements}} \\
    &= R_{AA'} (P_A^{(2m+1)}c_{2m}|2m+1,2m\rangle_{AB}|2m+1\rangle_{A'}+P_A^{(2m)}c_{2m+1}|2m,2m+1\rangle_{AB}|2m\rangle_{A'}) \\
    &= X_A^{(2m+1)}P_A^{(2m+1)}c_{2m}|2m+1,2m\rangle_{AB}|2m+1\rangle_{A'}+X_A^{(2m)}P_A^{(2m)}c_{2m+1}|2m,2m+1\rangle_{AB}|2m\rangle_{A'} \\
    &= c_{2m}|0,2m\rangle_{AB}|2m+1\rangle_{A'}+c_{2m+1}|0,2m+1\rangle_{AB}|2m\rangle_{A'} \\
    &= |0\rangle_A \otimes (c_{2m}|2m+1,2m\rangle_{A'B}+c_{2m+1}|2m,2m+1\rangle_{A'B}) \\
    &= |junk\rangle_A \otimes [\sigma_x]_m|\psi_{target}\rangle_{A'B}
    \end{aligned}
\end{equation}
Certification for subspace measurements $A_{x,m}$ can be shown similarly for the cases $x=0,2$.

\section{Robust SDI certification of all pure bipartite \textit{maximally} entangled states}
\label{robustness_appendix}
 We will adopt the assemblage based robust SDI certification methodology introduced in \cite{self_testing_epr} for proving our robust certification result for maximally entangled target state $|\Bar{\psi}_{target}^{max}\rangle := \frac{1}{\sqrt{d}}\sum_{i=0}^{d-1}|ii\rangle$. 
 First let us setup the reference experiment that will be studied in the rest of the section. Ideally we want the source to produce the reference state $|\Bar{\psi}_{target}^{max}\rangle = \frac{1}{\sqrt{d}}\sum_{i=0}^{d-1}|ii\rangle$ with $\Bar{\rho} = |\Bar{\psi}_{target}^{max}\rangle \langle \Bar{\psi}_{target}^{max}|$ and Alice to perform the reference measurements as given in subsection \ref{ideal_measurements} with reference projectors $\Bar{M}_{a|x}$ for $a\in \{0,\dots,d-1\}$ and $x\in \{0,1,2\}$. The reference assemblages $\{\Bar{\sigma}_{a|x}\}_{a,x}$ generated on Bob's side are as given in subsection \ref{assemblage_structure} with $\Bar{\sigma}_{a|x} = tr_A((\Bar{M}_{a|x} \otimes I)\Bar{\rho}_{AB})$. \\
 
 The physical set of assemblages generated on Bob's side are denoted as $\sigma_{a|x}$ and the physical reduced state on Bob's side as $\rho_B$, further we consider the purification of the reduced state $\rho_B$ as $|\psi_{AB}\rangle$ for formulating and proving our robust certification results. Before proving our main result, let us describe three lemmas which will be useful in the proof.
 
 \begin{lemma}
 \label{norm_lemma}
 [Borrowed from \cite{self_testing_epr} Lemma 2] For any 2 vectors $|u\rangle, \, |v\rangle$ with $\lVert |u\rangle \rVert \leq 1$ and $\lVert |v\rangle \rVert \leq 1$, where $ \lVert |u\rangle \rVert = \sqrt{\langle u|u\rangle}$, if  $\lVert |u\rangle - |v\rangle \rVert \leq \eta \leq 1$, then for another vector $|t\rangle$ with $\lVert |t\rangle \rVert \leq \beta$, $\lVert (|u\rangle - |v\rangle)\langle t| \rVert_1 \leq \beta \eta$ and $\lVert |t\rangle (\langle u| - \langle v|) \rVert_1 \leq \beta \eta$
 \end{lemma}
 
 \begin{lemma}
 \label{shifting_projector}
 If $\lVert \sigma_{a|x} - \Bar{\sigma}_{a|x} \rVert_1 \leq \epsilon \quad \forall \, a,x$ and $\dis(\rho_B, \Bar{\rho}_B) \leq \epsilon$ then
 \begin{equation}
     \lVert (U_A M_{a|x} \otimes I_B)|\psi\rangle_{AB} - (U_A \otimes \Bar{M}_{a|x}^B) |\psi\rangle_{AB} \rVert \leq 2\sqrt{\epsilon} \quad \forall \, a,x
 \end{equation}
 where $\Bar{M}_{a|x}^B$ denotes the reference projector $\Bar{M}_{a|x}$ acting on Bob's $d$-dimensional Hilbert space and $U_A$ is some unitary operator on Alice's side.
 \end{lemma}
 \begin{proof}
\begin{equation}
\begin{aligned}
    & \lVert (U_A M_{a|x} \otimes I_B)|\psi\rangle_{AB} - (U_A \otimes \Bar{M}_{a|x}^B) |\psi\rangle_{AB} \rVert \\
    &= \sqrt{\langle \psi_{AB}| M_{a|x} \otimes I_B |\psi\rangle_{AB} + \langle \psi_{AB}| I_A \otimes \Bar{M}_{a|x}^B |\psi\rangle_{AB} - 2\langle \psi_{AB}| M_{a|x} \otimes \Bar{M}_{a|x}^B |\psi\rangle_{AB} } \\
    &= \sqrt{Tr_B(\sigma_{a|x}) + Tr_B(\Bar{M}_{a|x}^B \rho_B) - 2Tr_B(\Bar{M}_{a|x}^B\sigma_{a|x}) } \\
    &\leq \sqrt{\frac{2}{d} + 2\epsilon - 2Tr_B(\Bar{M}_{a|x}^B\sigma_{a|x})} \\
    & \mbox{follows from $|Tr_B(\sigma_{a|x} - \Bar{\sigma}_{a|x})| \leq  \lVert \sigma_{a|x} - \Bar{\sigma}_{a|x} \rVert_1 \leq  \epsilon$ ; $|Tr_B(\Bar{M}_{a|x}^B(\rho_B - \Bar{\rho}_B))| \leq \dis(\rho_B, \Bar{\rho}_B) \leq \epsilon$ and} \\
    & \mbox{since we are dealing with maximally entangled qudit states, $Tr_B(\Bar{\sigma}_{a|x})= Tr_B(\Bar{M}_{a|x}^B\Bar{\rho}_B) = \frac{1}{d} \quad \forall \, a,x$} \\
    &\leq \sqrt{\frac{2}{d} + 2\epsilon - 2\left( \frac{1}{d} - \epsilon \right)} = 2\sqrt{\epsilon} \\
\end{aligned}
\end{equation}
since $|Tr_B(\Bar{M}_{a|x}^B(\sigma_{a|x} - \Bar{\sigma}_{a|x}))| = |Tr_B(\Bar{M}_{a|x}^B(\sigma_{a|x})) - \frac{1}{d}| \leq \lVert \sigma_{a|x} - \Bar{\sigma}_{a|x} \rVert_1 \leq  \epsilon \implies Tr_B(\Bar{M}_{a|x}^B(\sigma_{a|x})) \geq \frac{1}{d} - \epsilon $ 
 \end{proof}
 
 \begin{lemma}
 \label{shifting_unitary}
 If $\lVert \sigma_{a|x} - \Bar{\sigma}_{a|x} \rVert_1 \leq \epsilon \quad \forall \, a,x$ and $\dis(\rho_B, \Bar{\rho}_B) \leq \epsilon$ then
 \begin{gather*}
      \lVert (U_A X_{A,m}^u \otimes I_B)|\psi\rangle_{AB} - (U_A \otimes \Bar{X}_{B,m}^u) |\psi\rangle_{AB} \rVert \leq 4\sqrt{\epsilon} \\
     \lVert (U_A' Y_{A,m}^u \otimes I_B)|\psi\rangle_{AB} - (U_A' \otimes \Bar{Y}_{B,m}^u) |\psi\rangle_{AB} \rVert \leq 4\sqrt{\epsilon}    
 \end{gather*}

 where $\Bar{X}_{B,m}^u (\Bar{Y}_{B,m}^u)$ denotes the reference unitary $\Bar{X}_{A,m}^u (\Bar{Y}_{A,m}^u)$ acting on Bob's $d$-dimensional Hilbert space and $U_A\, (U_A')$ is some unitary operator on Alice's side.
 \end{lemma}
 \begin{proof}
    The proof follows closely that of \autoref{shifting_projector} where instead of shifting the projector $M_{a|x}$, we shift the unitary $X_{A,m}^u (Y_{A,m}^u)$ to Bob's side. Recall from subsection \ref{unitaries_def} that we defined; 
    \begin{gather*}
        X_{A,m}^u := A^u_{1,m} := \mathds{1}-\mathds{1}_m^{A_1}+A_{1,m} = \mathds{1} - 2\Pi_{2m+1}^{A_1} \\
        Y_{A,m}^u := A^u_{2,m} := \mathds{1}-\mathds{1}_m^{A_2}+A_{2,m} = \mathds{1} - 2\Pi_{2m+1}^{A_2}
    \end{gather*}
     where $\Pi_{2m+1}^{A_1}$ and $\Pi_{2m+1}^{A_2}$ are the projectors $M_{a|x}$ for $a=2m+1$ and $x=1,2$ respectively. Then we have,
    \begin{equation}
    \begin{aligned}
    & \lVert (U_A X_{A,m}^u \otimes I_B)|\psi\rangle_{AB} - (U_A \otimes \Bar{X}_{B,m}^u) |\psi\rangle_{AB} \rVert \\
    &= \lVert (U_A (\mathds{1} - 2M_{2m+1|1}) \otimes I_B)|\psi\rangle_{AB} - (U_A \otimes (\mathds{1} - 2\Bar{M}_{2m+1|1}^B)) |\psi\rangle_{AB} \rVert \\
    &= 2\lVert (U_A \otimes \Bar{M}_{2m+1|1}^B) |\psi\rangle_{AB} - (U_A M_{2m+1|1} \otimes I_B)|\psi\rangle_{AB} \rVert \\
    &\leq 4\sqrt{\epsilon} \quad \mbox{using \autoref{shifting_projector}}
    \end{aligned}
    \end{equation}
    The proof for the case $Y_{A,m}^u$ proceeds similarly.
 \end{proof}
 
 We will utilise these lemmas heavily in the proof of Result \ref{robust_result} which proceeds as;
 \begin{proof}
    Recall that Bob has experimentally verified that,
    \begin{gather*}
        \lVert \sigma_{a|x} - \Bar{\sigma}_{a|x} \rVert_1 \leq \epsilon \quad \forall \, a,x \mbox{ and } 
        \dis(\rho_B, \Bar{\rho}_B) \leq \epsilon
    \end{gather*}
    Let us start by first bounding the distance between the states $|\Phi\rangle := \Phi|\psi\rangle_{AB}|0\rangle_{A'}$ and $\rho_{junk} \otimes \Bar{\rho}_{A'B}$ where $\rho_{junk} = |junk\rangle \langle junk|_A$ and $\Bar{\rho}_{A'B} = |\Bar{\psi}_{target}^{max}\rangle \langle \Bar{\psi}_{target}^{max}|_{A'B}$. As given in equation \eqref{b21}, the effect of applying the local isometry $\Phi$ on Alice's side is $\Phi|\psi\rangle_{AB}|0\rangle_{A'} = \sum_j X_A^{(j)} P_A^{(j)} |\psi\rangle_{AB}|j\rangle_{A'}$ such that $P_A^{(j)} := M_{j|0}$ and $X_A^{(j)}$ are constructed as shown in \eqref{big_unitary};
    \begin{equation}
    X_A^{(j)}=
    \begin{cases}
    \mathds{1}, \mbox{ for } j=0 \\
    X_{A,0}^u Y_{A,0}^u X_{A,1}^u Y_{A,1}^u \ldots X_{A,m-1}^u Y_{A,m-1}^u X_{A,m}^u, \mbox{ for } j=2m+1 \\
    X_{A,0}^u Y_{A,0}^u X_{A,1}^u Y_{A,1}^u \ldots X_{A,m-1}^u Y_{A,m-1}^u, \mbox{ for } j=2m \\
    \end{cases}
    \end{equation}
    Note that the physical unitaries $X_A^{(j)}$ and projectors $P_A^{(j)}$ do not satisfy the condition \eqref{lem_condition}, however the local isometry $\Phi$ can always be constructed such that equation \eqref{b21} is valid. Then, 
    \begin{equation}
    \label{c5}
    \begin{aligned}
        & \snorm*{|\Phi\rangle\langle \Phi| - \rho_{junk} \otimes \Bar{\rho}_{A'B} } \\
        &= \snorm*{\left(\sum_{j=0}^{d-1} X_A^{(j)} M_{j|0} |\psi\rangle_{AB}|j\rangle_{A'}\right) \langle \Phi| - \rho_{junk} \otimes \Bar{\rho}_{A'B}} \\
        &\leq 2d\sqrt{\epsilon} + \snorm*{\left(\sum_{j=0}^{d-1} X_A^{(j)} \otimes \Bar{M}_{j|0}^B |\psi\rangle_{AB}|j\rangle_{A'}\right) \langle \Phi| - \rho_{junk} \otimes \Bar{\rho}_{A'B}} \\
        & \mbox{using \autoref{norm_lemma} and \autoref{shifting_projector} and the fact that $\norm{|\Phi\rangle} = 1$ } \\
    \end{aligned}
    \end{equation}
    Now see that each term of the sum is of the form $X_A^{(j)} \otimes \Bar{M}_{j|0}^B |\psi\rangle_{AB}|j\rangle_{A'}$ where $X_A^{(j)}$ is composed of $j$ unitaries $X_{A,m}^u$ which can be shifted to Bob's side using \autoref{shifting_unitary} such that,
    \begin{equation}
        \begin{aligned}
            \snorm*{X_A^{(j)} \otimes \Bar{M}_{j|0}^B |\psi\rangle_{AB}|j\rangle_{A'} - I_A \otimes \Bar{M}_{j|0}^B (\Bar{X}_B^{(j)})^{\dagger} |\psi\rangle_{AB}|j\rangle_{A'}} \leq 4j\sqrt{\epsilon}
        \end{aligned}
    \end{equation}
 
    Continuing from \eqref{c5},
    \begin{equation}
    \label{c7}
        \begin{aligned}
            & 2d\sqrt{\epsilon} + \snorm*{\left(\sum_{j=0}^{d-1} X_A^{(j)} \otimes \Bar{M}_{j|0}^B |\psi\rangle_{AB}|j\rangle_{A'}\right) \langle \Phi| - \rho_{junk} \otimes \Bar{\rho}_{A'B}} \\
            &\leq 2d\sqrt{\epsilon} + \sum_{j=0}^{d-1} 4j\sqrt{\epsilon} + \snorm*{\left(\sum_{j=0}^{d-1} I_A \otimes \Bar{M}_{j|0}^B (\Bar{X}_B^{(j)})^{\dagger} |\psi\rangle_{AB}|j\rangle_{A'}\right) \langle \Phi| - \rho_{junk} \otimes \Bar{\rho}_{A'B}} \\
            &= 2d^2\sqrt{\epsilon} + \snorm*{\left(\sum_{j=0}^{d-1} I_A \otimes \Bar{M}_{j|0}^B (\Bar{X}_B^{(j)})^{\dagger} |\psi\rangle_{AB}|j\rangle_{A'}\right) \langle \Phi| - \rho_{junk} \otimes \Bar{\rho}_{A'B}} \\
        \end{aligned}
    \end{equation}
    Now we apply the same strategy to shift the projectors $P_A^{(j)}$ and the unitary operators $X_A^{(j)}$ to Bob's side in $\langle\Phi|$ but before that we need $\norm*{\sum_{j=0}^{d-1} I_A \otimes \Bar{M}_{j|0}^B (\Bar{X}_B^{(j)})^{\dagger} |\psi\rangle_{AB}|j\rangle_{A'}}$:
    \begin{equation}
        \begin{aligned}
            \norm*{\sum_{j=0}^{d-1} I_A \otimes \Bar{M}_{j|0}^B (\Bar{X}_B^{(j)})^{\dagger} |\psi\rangle_{AB}|j\rangle_{A'}}
            &= \sqrt{\sum_{j=0}^{d-1} \langle\psi_{AB}| I_A \otimes \Bar{X}_B^{(j)} \Bar{M}_{j|0}^B (\Bar{X}_B^{(j)})^{\dagger}|\psi\rangle_{AB}} \\
            &= \sqrt{\sum_{j=0}^{d-1} \langle\psi_{AB}| I_A \otimes \Bar{M}_{0|0}^B|\psi\rangle_{AB}} \\
            &= \sqrt{d(Tr_B(\Bar{M}_{0|0}^B \rho_B))}\\
            &\leq \sqrt{d\left(\epsilon + \frac{1}{d}\right)} \\
            &\leq 1 + \frac{d\epsilon}{2} 
        \end{aligned}
    \end{equation}
    Continuing from \eqref{c7},
    \begin{equation}
        \begin{aligned}
            & 2d^2\sqrt{\epsilon} + \snorm*{\left(\sum_{j=0}^{d-1} I_A \otimes \Bar{M}_{j|0}^B (\Bar{X}_B^{(j)})^{\dagger} |\psi\rangle_{AB}|j\rangle_{A'}\right) \langle \Phi| - \rho_{junk} \otimes \Bar{\rho}_{A'B}} \\
            &\leq 2d^2\sqrt{\epsilon} + 2d^2\sqrt{\epsilon}\left(1 + \frac{d\epsilon}{2} \right) + \\
            &\quad \snorm*{\left(\sum_{j=0}^{d-1} I_A \otimes \Bar{M}_{j|0}^B (\Bar{X}_B^{(j)})^{\dagger} |\psi\rangle_{AB}|j\rangle_{A'}\right) \left(\sum_{j=0}^{d-1} \langle\psi|_{AB} I_A \otimes \Bar{X}_B^{(j)} \Bar{M}_{j|0}^B \langle j|_{A'}\right) - \rho_{junk} \otimes \Bar{\rho}_{A'B}} \\
            &= 4d^2\sqrt{\epsilon} + d^3\epsilon\sqrt{\epsilon} + \snorm*{\left(\sum_{j=0}^{d-1} |j_B\rangle\langle0_B |\psi\rangle_{AB}|j\rangle_{A'}\right) \left(\sum_{j=0}^{d-1} \langle\psi_{AB} |0_B\rangle\langle j_B| \langle j|_{A'}\right) - \rho_{junk} \otimes \Bar{\rho}_{A'B}} \\
            &= 4d^2\sqrt{\epsilon} + d^3\epsilon\sqrt{\epsilon} + \snorm*{\left( \langle0_B |\psi\rangle_{AB}\sqrt{d}|\Bar{\psi}_{target}^{max}\rangle_{A'B}\right) \left( \langle\psi_{AB}|0_B\rangle \sqrt{d}\langle \Bar{\psi}_{target}^{max}|_{A'B} \right) - \rho_{junk} \otimes \Bar{\rho}_{A'B}} \\
            &= 4d^2\sqrt{\epsilon} + d^3\epsilon\sqrt{\epsilon} + \snorm*{d\langle 0_B|\psi_{AB}\rangle\langle\psi_{AB}|0_B\rangle \otimes \Bar{\rho}_{A'B} - \rho_{junk} \otimes \Bar{\rho}_{A'B}} \\
            &= 4d^2\sqrt{\epsilon} + d^3\epsilon\sqrt{\epsilon} + \snorm*{d\langle 0_B|\psi_{AB}\rangle\langle\psi_{AB}|0_B\rangle - \rho_{junk}} \\
            &\leq 4d^2\sqrt{\epsilon} + d^3\epsilon\sqrt{\epsilon} + d\epsilon
        \end{aligned}
    \end{equation}
    where for obtaining the last inequality we assigned $\rho_{junk} = |junk_A\rangle\langle junk_A|$ to be the normalized state proportional to $\langle 0_B|\psi_{AB}\rangle$ i.e. $|junk_A\rangle = \beta^{\frac{-1}{2}}\langle 0_B|\psi_{AB}\rangle$ with $\beta = \langle \psi_{AB}| 0_B\rangle \langle 0_B|\psi_{AB}\rangle$ and made the observation that $|Tr_B(|0_B\rangle\langle0_B|\rho_B - |0_B\rangle\langle0_B|\Bar{\rho}_B)| = |\beta - \frac{1}{d}| \leq \dis(\rho_B, \Bar{\rho}_B) \leq \epsilon$.

    Having shown that $\dis(|\Phi\rangle\langle \Phi|, \rho_{junk} \otimes \Bar{\rho}_{A'B})=\frac{1}{2}\snorm{|\Phi\rangle\langle \Phi| - \rho_{junk} \otimes \Bar{\rho}_{A'B}} \leq \frac{1}{2}(4d^2\sqrt{\epsilon} + d^3\epsilon\sqrt{\epsilon} + d\epsilon)$, now we proceed to show the robustness result with measurements, i.e. we provide an upper bound for the distance
    \begin{equation}
    \label{c10}
        \snorm{|\Phi,M_{a|x}\rangle\langle \Phi,M_{a|x}| - \rho_{junk} \otimes (\Bar{M}_{a|x}\otimes I_B)|\Bar{\psi}_{target}^{max}\rangle \langle \Bar{\psi}_{target}^{max}|(\Bar{M}_{a|x}\otimes I_B)}
    \end{equation}
    where $|\Phi,M_{a|x}\rangle := \Phi(M_{a|x}\otimes I_B)|\psi\rangle_{AB}|0\rangle_{A'}$ \\
    
    Let us first consider the case for $x=0$, then $|\Phi,M_{a|0}\rangle = \sum_j X_A^{(j)} P_A^{(j)} (M_{a|0}\otimes I_B) |\psi\rangle_{AB}|j\rangle_{A'} $ where $P_A^{(j)} = M_{j|0}$ and for projectors we have $M_{j|0}M_{a|0} = \delta_{a,j}M_{a|0}$. Then equation \eqref{c10} becomes:
    \begin{equation}
        \begin{aligned}
            & \snorm{|\Phi,M_{a|0}\rangle\langle \Phi,M_{a|0}| - \rho_{junk} \otimes (\Bar{M}_{a|0}\otimes I_B)|\Bar{\psi}_{target}^{max}\rangle \langle \Bar{\psi}_{target}^{max}|(\Bar{M}_{a|0}\otimes I_B)} \\
            &= \snorm*{\left( X_A^{(a)} P_A^{(a)} |\psi\rangle_{AB}|a\rangle_{A'}\right)\left( \langle \psi|_{AB} P_A^{(a)} (X_{A}^{(a)})^{\dagger} \langle a|_{A'} \right) - \rho_{junk} \otimes \frac{1}{d} |aa\rangle\langle aa|_{A'B} } \\
            &\leq 4a\sqrt{\epsilon} + 2\sqrt{\epsilon} + (4a\sqrt{\epsilon} + 2\sqrt{\epsilon})\norm{I_A \otimes \Bar{M}_{j|0}^B (\Bar{X}_B^{(j)})^{\dagger}|\psi\rangle_{AB}|a\rangle_{A'}} \, + \\
            &\quad \snorm*{\left( I_A \otimes \Bar{M}_{a|0}^B (\Bar{X}_B^{(a)})^{\dagger} |\psi\rangle_{AB}|a\rangle_{A'}\right)\left( \langle\psi|_{AB} I_A\otimes\Bar{X}_{B}^{(a)}\Bar{M}_{a|0}^B \langle a|_{A'} \right) - \rho_{junk} \otimes \frac{1}{d} |aa\rangle\langle aa|_{A'B} }\\  
            &\leq 2(4a\sqrt{\epsilon} + 2\sqrt{\epsilon}) + \snorm*{\langle 0_B|\psi_{AB}\rangle\langle\psi_{AB}|0_B\rangle - \frac{\rho_{junk}}{d}} \\
            &\leq (8a+4)\sqrt{\epsilon} + \epsilon \quad \mbox{where $a \in \{0,\ldots, d-1\}$ }\\
            &\leq (8d-4)\sqrt{\epsilon} + \epsilon
        \end{aligned}
    \end{equation}
    Next consider the case for $x=1$ and $a=2m$, the case for $a=2m+1$ yields the same bound, then equation \eqref{c10} becomes:
    \begin{equation}
        \begin{aligned}
            & \snorm{|\Phi,M_{a|1}\rangle\langle \Phi,M_{a|1}| - \rho_{junk} \otimes (\Bar{M}_{a|1}\otimes I_B)|\Bar{\psi}_{target}^{max}\rangle \langle \Bar{\psi}_{target}^{max}|(\Bar{M}_{a|1}\otimes I_B)} \\
            &= \snorm*{\left( \sum_{j=0}^{d-1} X_A^{(j)} P_A^{(j)} M_{a|1}|\psi\rangle_{AB}|j\rangle_{A'}\right)\langle \Phi,M_{a|1}| - \rho_{junk} \otimes \frac{1}{d}|+_{2m}+_{2m}\rangle \langle+_{2m}+_{2m}|_{A'B} } \\
            & \mbox{where $|+_{2m}\rangle = \frac{|2m\rangle + |2m+1\rangle}{\sqrt{2}}$; repeatedly using \autoref{shifting_projector} and \autoref{shifting_unitary} we get} \\
            &\leq \sum_{j=0}^{d-1}(2\sqrt{\epsilon}+2\sqrt{\epsilon}+4j\sqrt{\epsilon})\, + \\
            &\quad \snorm*{\left( \sum_{j=0}^{d-1} \Bar{M}_{a|1}^B \Bar{M}_{j|0}^B (\Bar{X}_B^{(j)})^{\dagger} |\psi\rangle_{AB}|j\rangle_{A'}\right)\langle \Phi,M_{a|1}| - \rho_{junk} \otimes \frac{1}{d}|+_{2m}+_{2m}\rangle \langle+_{2m}+_{2m}|_{A'B} } \\
            &\leq 4(d^2+d)\sqrt{\epsilon}\, + \\ 
            &\quad \snorm*{\bigg(\langle0_B|\psi_{AB}\rangle|+_{2m}+_{2m}\rangle_{A'B}\bigg) \bigg(\langle\psi_{AB}|0_B\rangle\langle+_{2m}+_{2m}|_{A'B}\bigg) - \rho_{junk} \otimes \frac{1}{d}|+_{2m}+_{2m}\rangle \langle+_{2m}+_{2m}|_{A'B} } \\
            &= 4(d^2+d)\sqrt{\epsilon} + \snorm*{\langle 0_B|\psi_{AB}\rangle\langle\psi_{AB}|0_B\rangle - \frac{\rho_{junk}}{d}}\\
            &\leq 4(d^2+d)\sqrt{\epsilon} + \epsilon
        \end{aligned}
    \end{equation}
    Finally, the case for $x=2$ proceeds similarly giving the same bound.

\end{proof}
\bibliography{steering}

\end{document}